\documentclass[a4paper,USenglish,thm-restate]{lipics-v2021}
\usepackage{algorithm}
\usepackage{algpseudocode}

\newif\IFdraft
\drafttrue

\usepackage{amsmath, amsthm, amssymb}
\usepackage[nameinlink]{cleveref}

\usepackage{xcolor}
\usepackage{xspace}
\usepackage{todonotes}
\usepackage{customdice}
\usepackage{booktabs}

\nolinenumbers

\newcommand{\LOCAL}{\ensuremath{\mathsf{LOCAL}}\xspace}
\newcommand{\CONGEST}{\ensuremath{\mathsf{CONGEST}}\xspace}

\newcommand{\cR}{\mathcal{R}}

\title{Towards Optimal Distributed Edge Coloring with Fewer Colors}

\author{Manuel Jakob}{TU Graz, Austria}{m.jakob@tugraz.at}{https://orcid.org/0009-0009-0229-0287}{}
\author{Yannic Maus}{TU Graz, Austria}{yannic.maus@tugraz.at}{https://orcid.org/0000-0003-4062-6991}{}
\author{Florian Schager}{TU Graz, Austria}{florian.schager@tugraz.at}{https://orcid.org/0009-0009-3923-051X}{}
\authorrunning{M. Jakob, Y. Maus and F. Schager}
\Copyright{Manuel Jakob, Yannic Maus and Florian Schager}
\date{\today}
\acknowledgements{This research was funded in whole or in part by the Austrian Science Fund (FWF) \url{https://doi.org/10.55776/P36280}, \url{https://doi.org/10.55776/I6915}. For open access purposes, the author has applied a CC BY public copyright license to any author-accepted manuscript version arising from this submission.}

\newcommand{\poly}{\mathrm{poly}}

\begin{document}
\maketitle

\begin{abstract}
There is a huge difference in techniques and runtimes of distributed algorithms for problems that can be solved by a sequential greedy algorithm and those that cannot. 
A prime example of this contrast appears in the edge coloring problem: while $(2\Delta-1)$-edge coloring---where $\Delta$ is the maximum degree---can be solved in $\mathcal{O}(\log^{\ast}(n))$ rounds on constant-degree graphs, the seemingly minor reduction to $(2\Delta-2)$ colors leads to an $\Omega(\log n)$ lower bound [Chang, He, Li, Pettie \& Uitto, SODA'18]. 
Understanding this sharp divide between very local problems and inherently more global ones remains a central open question in distributed computing and it is a core focus of this paper.

As our main contribution we design a deterministic distributed $\mathcal{O}(\log n)$-round reduction from the $(2\Delta-2)$-edge coloring problem to the much easier $(2\Delta-1)$-edge coloring problem. 
This reduction is optimal, as the $(2\Delta-2)$-edge coloring problem admits an  $\Omega(\log n)$ lower bound that even holds on the class of constant-degree graphs, whereas the $2\Delta-1$-edge coloring problem can be solved in $\mathcal{O}(\log^{\ast}n)$ rounds. 
By plugging in the $(2\Delta-1)$-edge coloring algorithms from [Balliu, Brandt, Kuhn \& Olivetti, PODC'22] running in $\mathcal{O}(\log^{12}\Delta + \log^{\ast} n)$ rounds, we obtain an optimal runtime of $\mathcal{O}(\log n)$ rounds as long as $\Delta = 2^{\mathcal{O}(\log^{1/12} n)}$. 
Previously, such an optimal algorithm was only known for the class of constant-degree graphs [Brandt, Maus, Narayanan, Schager \& Uitto, SODA'25]. 
Furthermore, on general graphs our reduction improves the runtime from $\widetilde{\mathcal{O}}(\log^3 n)$ to $\widetilde{\mathcal{O}}(\log^{5/3} n)$.

In addition, we also obtain an optimal $\mathcal{O}(\log \log n)$-round randomized reduction of $(2\Delta - 2)$-edge coloring to $(2\Delta - 1)$-edge coloring. 
This leads to a $\widetilde{\mathcal{O}}(\log^{5/3} \log n)$-round $(2\Delta-2$)-edge coloring algorithm, which beats the (very recent) previous state-of-the-art taking $\widetilde{\mathcal{O}}(\log^{8/3}\log n)$ rounds from [Bourreau, Brandt \& Nolin, STOC'25].

Lastly, we obtain an $\mathcal{O}(\log_\Delta n)$-round reduction from the $(2\Delta-1)$-edge coloring, albeit to the somewhat harder maximal independent set (MIS) problem. 

\end{abstract}
\newpage

\thispagestyle{empty}
\tableofcontents
\newpage

\pagenumbering{arabic}
\setcounter{page}{1}

\section{Introduction}

There is a huge difference in the techniques and runtimes of distributed algorithms for locally checkable problems that can be solved by a sequential greedy algorithm and those that cannot. 
For example, on bounded degree graphs, the former problems have distributed complexity $\Theta(\log^{\ast}n)$ \cite{ColeVishkin,Linial1,SimpleAlgorithms} while the latter require at least $\Omega(\log n)$ time \cite{SmallPalettes,ExponentialSeparation}. 
Understanding this sharp and potentially wide divide between very locally solvable problems and inherently more global ones remains a central open question in distributed computing. 

A prime example of this contrast---explained in detail next and forming the main topic of this paper---appears in the edge coloring problem. 
A \emph{$k$-edge coloring} of a graph $G = (V,E)$ is a function $\varphi: E \to \{ 1,\dots, k \} $ such that no two adjacent edges receive the same color. 
This problem \cite{SimpleAlgorithms,Deltalogn,MuchEasier,EdgeColorHyperMaxMatch,EdgeColoringPolylogarithmicInDelta,SublogarithmicEdgeColoring}, but also other \emph{greedy problems} like $(\Delta + 1)$-vertex coloring, maximal independent set, and maximal matching \cite{LubySTOC,ABI,LocalitySymmetryBreaking,VertexColoringSublinear,OptimalColoring,MISLowerBound,NDCompBreakthrough,ColoringWithoutNetDComp,Rounding,GreedyAlgorithm}, has been extensively studied in the field of distributed computing, particularly in the \LOCAL model.

In this model the communication network is abstracted as an undirected graph $G = (V,E)$.
Each vertex hosts a processor that may communicate with its neighbors in synchronous rounds.
In one round, each processor may send unbounded-size messages to all its neighbors in $G$ and perform unbounded local computations.
The goal is to minimize the number of synchronous rounds until every node outputs (a local part of) the solution. 
The related \CONGEST model additionally imposes a maximum message size of $\mathcal{O}(\log n)$ bits per message.

In order to make progress in many parts of the graph in parallel, many distributed algorithms crucially rely on the property that color choices made by other vertices do not affect the solvability on the remaining graph. 
While this property is satisfied for $k = 2\Delta - 1$ (hence we call this the \emph{greedy threshold}), for any $k < 2\Delta - 1$ there are partial $k$-edge colorings that cannot be extended to the entire graph \cite{OnlineEdgeColoring}.
On the other hand of the spectrum, it is known that every graph admits a $(\Delta + 1)$-edge coloring by a celebrated result of Vizing \cite{Vizing}.

The intuition that coloring with fewer colors is significantly more challenging---particularly in the distributed setting---is evidenced by strong lower bounds: 
Any deterministic edge coloring algorithm using less than $2\Delta - 1$ colors requires at least $\Omega(\log_\Delta n)$ rounds, and even randomized algorithms face an $\Omega(\log_{\Delta}\log n)$-round lower bound \cite{SmallPalettes}. 
On trees, where these lower bounds are established, there is even a deterministic algorithm that solves the greedy version of the problem in strongly sublogarithmic time, specifically in $\mathcal{O}(\log^{12/13}n)$ rounds \cite{SublogarithmicEdgeColoring}.
On bounded degree graphs, the gap is even larger as the greedy version has deterministic and randomized complexity  $\Theta(\log^{\ast}n)$, yet the lower bounds for algorithms that use fewer colors---$\Omega(\log n)$ deterministic and $\Omega(\log \log n)$ randomized---still apply.

Despite these challenges, there has been extensive work on designing increasingly faster deterministic and randomized algorithms for edge coloring with fewer than $2\Delta-1$ colors---both in the \LOCAL model \cite{Bernshteyn,VizingBoundedDegree,HSO,SmallPalettes,SmallPalettesJournal,3Delta/2,3Delta/2HMM}, the more restricted \CONGEST model \cite{3Delta/2,FastVertexSplitting}, and in other models such as streaming \cite{OnlineEdgeColoring,OnlineEdgeColoringTightBounds,GreedyAlgorithmNotOptimal}. 
The fastest known deterministic distributed algorithm that operates below the greedy threshold runs in $\widetilde{\mathcal{O}}(\log^3 n)$ rounds\footnote{We use the notation $\widetilde{\mathcal{O}}(f(n))$ to hide factors polylogarithmic in $f(n)$.} \cite{HSO}, while the fastest randomized algorithm, derived from a recent more general result on $\Delta$-vertex coloring, achieves a runtime of $\widetilde{\mathcal{O}}(\log^{8/3} \log n)$ rounds \cite{FasterDeltaColoring}.
See the related work section (\Cref{subsec:related_work}) for further results on edge coloring with fewer colors, and refer to \Cref{tab:edge_coloring} for an overview of existing algorithms, the number of colors they use, and their runtime.

This work aims to better understand the size and nature of the complexity gap between greedy-solvable problems and those below the greedy threshold. Understanding this latter class of problems is also central to the poorly understood regime in the complexity landscape of local graph problems, see, e.g., \cite{SublogarithmicLLL, TimeHierarchy, CongestLLL}.

\subsection{Our contributions}

As our main contribution, we show that the considerably harder $(2\Delta - 2)$-edge coloring problem can be efficiently reduced to its greedy cousin $(2\Delta - 1)$-edge coloring\footnote{We use $T_{2\Delta - 1}(n,\Delta)$ ($T_{2\Delta - 1}^{\dice{6}}(n,\Delta)$) to denote the runtime of a deterministic (randomized) $(2\Delta - 1)$-edge coloring algorithm and $T_{\mathrm{MIS}}(n,\Delta)$ to denote the runtime of a deterministic MIS algorithm for $n$-vertex graphs with maximum degree $\Delta$. } in $\mathcal{O}(\log n)$ rounds of the \LOCAL model.

\begin{restatable}{theorem}{detAlgo} \label{thm:detAlgo}
    There is a deterministic \LOCAL algorithm computing a $(2\Delta-2)$-edge coloring in any $n$-vertex graph with maximum degree $\Delta$ in $\mathcal{O}(\log n) + T_{2\Delta - 1}(n,\Delta-1)$ rounds.
\end{restatable}

Secondly, allowing the use of maximal independent set (MIS) as a subroutine improves the runtime of our reduction at the cost of reducing to a harder problem.

\begin{restatable}{theorem}{MISalgo}
    \label{thm:simple_algo}
    There is a deterministic \LOCAL algorithm computing a $(2\Delta - 2)$-edge coloring in any $n$-vertex graph with maximum degree $\Delta$ in $\mathcal{O}(\log_\Delta n) + T_{\mathrm{MIS}}(\Delta^2 \cdot n,\mathrm{poly}(\Delta))$ rounds.
\end{restatable}

While the latter result might sound like a strict improvement over the former---especially given the fact that the runtimes of the current state-of-the-art algorithms for $(2\Delta - 1)$-edge coloring and MIS on general graphs coincide---we argue that the former version is, in fact, the stronger result.
A classic reduction from Luby \cite{LubySTOC} shows that $(2\Delta - 1)$-edge coloring is no harder than MIS.
Consequently, the $\widetilde{\mathcal{O}}(\log^{5/3} n)$ deterministic algorithm for MIS \cite{GreedyAlgorithm} also solves the $(2\Delta - 1)$-edge coloring problem. 
On the other hand, there is a $\mathcal{O}(\log^{12} \Delta + \log^{\ast} n)$-round $(2\Delta - 1)$-edge coloring algorithm from \cite{EdgeColoringPolylogarithmicInDelta}, which provably cannot be achieved by any MIS algorithm due to the $\Omega(\min(\Delta, \log n / \log \log n))$ lower bound established in \cite{LowerBoundMIS}.
We obtain the following corollary.

\begin{restatable}{corollary}{detAlgoCor} \label{cor:detAlgo}
There are deterministic \LOCAL algorithms computing a $(2\Delta - 2)$-edge coloring in any $n$-vertex graph with maximum degree $\Delta$ within
    \begin{enumerate}
        \item either $\mathcal{O}(\log^{12}\Delta + \log n)$ rounds,
        \item or $\widetilde{\mathcal{O}}(\log^{5/3} n)$ rounds.
    \end{enumerate}
\end{restatable}

For the class of graphs with maximum degree $\Delta$ as large as $2^{\mathcal{O}(\log^{1/12}n)}$, the first point of this corollary yields an optimal runtime of $\mathcal{O}(\log n)$ rounds. 
Previously, an optimal runtime was only known for the class of bounded-degree graphs \cite{HSO}.
The second point significantly improves the upper bound on the complexity of $(2\Delta - 2)$-edge coloring on general graphs, which previously stood at $\widetilde{\mathcal{O}}(\log^3 n)$ rounds \cite{HSO}.
Additionally, we present an exponentially faster randomized reduction to the $(2\Delta-1)$-edge coloring problem. 

\begin{restatable}{theorem}{randAlgo} \label{thm:randAlgo}
    There is a randomized \LOCAL algorithm computing a $(2\Delta-2)$-edge coloring in any $n$-vertex graph with maximum degree $\Delta$ in $\mathcal{O}(\log \log n) + T_{2\Delta - 1}^{\dice{6}}(n,\Delta-1)$ rounds.
\end{restatable}

\Cref{thm:randAlgo} yields the following corollary by plugging in the algorithms of \cite{EdgeColoringPolylogarithmicInDelta} and \cite{GreedyAlgorithm}.

\begin{restatable}{corollary}{randAlgoCor}
    There are randomized distributed algorithms computing a $(2\Delta-2)$-edge coloring in any $n$-vertex graph with maximum degree $\Delta$ within
    \begin{enumerate}
        \item either $\mathcal{O}(\log^{12} \Delta + \log \log n)$ rounds,
        \item or $\widetilde{\mathcal{O}}(\log^{5/3}\log n)$ rounds of the \LOCAL model.
    \end{enumerate}
\end{restatable}

For the class of graphs with maximum degree $\Delta$ as large as $2^{\mathcal{O}(\log^{1/12}\log n)}$, the first point of this corollary yields an optimal runtime of $\mathcal{O}(\log \log n)$ rounds. 
The second point beats the (very recent) previous state-of-the-art for randomized $(2\Delta - 2)$-edge coloring on general graphs, which previously stood at $\widetilde{\mathcal{O}}(\log^{8/3}\log n)$ rounds \cite{FasterDeltaColoring}.

\subsection{Our technique in a nutshell}

To illustrate the difficulty of the $(2\Delta-2)$-edge coloring problem, let us first try to design an $\mathcal{O}(1)$-round reduction to the MIS and the $(2\Delta-1)$-edge coloring problem; recall that such a reduction provably cannot exist. 
One approach is to use the MIS algorithm to compute a clustering of the graph, e.g., by computing an MIS of the power graph $G^2$ and letting each vertex join the cluster of the closest node in the MIS; see \Cref{fig:clusters} for an example. 
As a result, we can partition the edge set of $G$ into edges $E_{\mathrm{intra}}$ (black) within the clusters and inter-cluster edges $E_{\mathrm{inter}}$ (turquoise) that go between the different clusters. 
We could simply color the edges in $E_{\mathrm{intra}}$ independently and in parallel for each cluster, as the edges of different clusters are non-adjacent. 
Due to the constant diameter of each cluster, this can be done in $O(1)$ rounds and even using the existentially optimal number of colors, i.e., $\Delta+1$ colors. 
However, now the problem of coloring $E_{\mathrm{intra}}$ is no easier than the original $(2\Delta-2)$-edge coloring problem as it essentially corresponds to coloring a general graph but with some colors---the ones used by adjacent edges in $E_{\mathrm{intra}}$---being  forbidden for certain edges. 
In fact, it might not even be possible to complete the coloring at all. 

\begin{figure}[ht!]
    \captionsetup{justification=centering}
    \centering
    \includegraphics{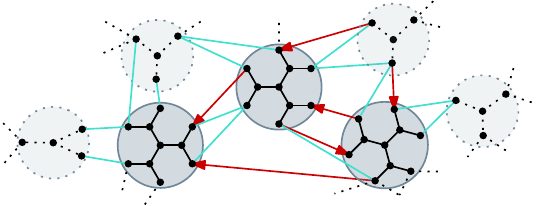}
    \caption{A matching of the intercluster edges that assigns two edges to each clusters. The matching edges are marked red. Arrows point towards the cluster the edge belongs to.}
    \label{fig:clusters}
\end{figure}

Instead, in our algorithm, we first color the edges between the clusters and then try to complete it within the clusters. 
First, observe that the graph induced by $E_{\mathrm{inter}}$ has maximum degree $\Delta'=\Delta-1$, as each node has at least one adjacent edge within a cluster. 
Thus, we can color $E_{\mathrm{inter}}$ with one instance of greedy-type edge coloring, using only $2\Delta'-1=2\Delta-3$ colors. 
The main challenge that we address in this paper is how to ensure that the coloring can be extended to the edges within the clusters. 

For the rest of the this technical overview we assume that the graph is $\Delta$-regular with a sufficiently high girth such that each cluster looks like a tree; our full proof does not require this assumption.
In addition, let us assume for now that each cluster is a perfect $\Delta$-ary tree, that is, all leaves are at the same level.
Then, a simplified version of our main technical contribution gives a useful condition under which we can complete the coloring.

\begin{lemma}[Simplified version of \Cref{lem:colorful_leaves}]
    Consider a $(2\Delta - 2)$-edge colored graph except for a perfect $\Delta$-ary subtree $T \subseteq G$.
    If the number of distinct colors appearing on the incident edges of any leaf in $T$ combined is at least $\Delta$, then the coloring can be extended to $T$.
\end{lemma}

For the remainder of this paper we will refer to the above condition as the \emph{colorful condition}.
To build some intuition for it, consider the simplest possible case, where the set of uncolored edges forms a star graph $S \subseteq G$. 
The only case, in which $\varphi$ cannot be extended to $S$ is if each edge in $S$ has the exact same set of $\Delta - 1$ colors appearing among its adjacent edges, which we illustrate in \Cref{fig:star_graph}. 
By setting up a bipartite graph with available colors on one side and edges on the opposite side, an application of Hall's marriage theorem shows that this is indeed the only bad case.
In the proof of \Cref{lem:colorful_leaves}, we generalize this condition from uncolored stars to  uncolored trees. 

\begin{figure}[ht!]
    \captionsetup{justification=centering}
    \centering
    \begin{subfigure}[t]{0.49 \textwidth}
        \centering
        \includegraphics{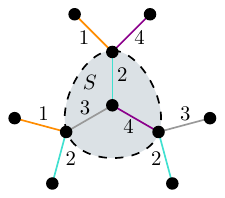}
        \caption{This coloring can be extended.}
    \end{subfigure}
    \begin{subfigure}[t]{0.49 \textwidth}
        \centering
        \includegraphics{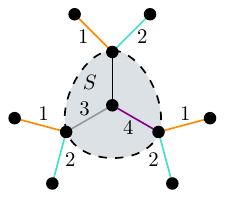}
        \caption{This coloring cannot be extended.}
        \label{fig:forbidden_configuration}
    \end{subfigure}
    \caption{Extendability of $(2\Delta - 2)$-colorings for $\Delta = 3$.}
    \label{fig:star_graph}
\end{figure}

Next, we turn our attention to how the partial edge coloring of $E_{\mathrm{inter}}$  can be modified in order to satisfy the \emph{colorful condition} for every cluster of uncolored edges.
The core observation is that each leaf in the cluster already has $\Delta - 1$ incident edges in $E_{\mathrm{inter}}$, whose colors must all be distinct. 
Therefore, if we can manage to keep the set of colors fixed for one leaf and modify at most one color at another leaf, the colorful condition is already satisfied.
Notably, our partial coloring of $E_{\mathrm{inter}}$ only uses $2\Delta - 3$ of the $2\Delta - 2$ available colors.
As a result, every edge already has at least one free color it can switch to if needed.

The main difficulty now is to construct a set of edges such that each cluster can perform these changes in parallel and without causing any conflicts.
Hence, our desired set of edges needs to satisfy two key properties: 
(1) Each cluster needs exclusive access to two edges in its $1$-hop neighborhood. 
(2) The set of all such edges needs to be a matching.
Computing a greedy maximal matching satisfies (2), but not necessarily (1). 
Therefore, we modify a maximal matching $M \subseteq E_{\mathrm{inter}}$ by letting each cluster send proposals to matching edges in its $2$-hop neighborhood.

\begin{figure}[htbp]
    \centering
    \begin{subfigure}[c]{0.31\textwidth}
        \centering
        \includegraphics{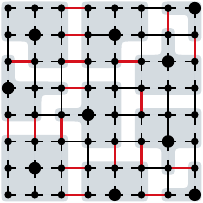}
        \caption{The intracluster edges $E_{\mathrm{intra}}$ are dashed. We first compute a maximal matching $M \subseteq E_{\mathrm{inter}}$, which is highlighted red.}
        \label{fig:maximal_matching}
    \end{subfigure}
    \hfill 
    \begin{subfigure}[c]{0.31\textwidth}
        \centering
        \includegraphics{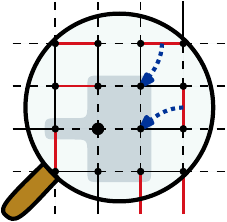}
        \caption{Next, we assign two edges in $M$ to each cluster. This is the only ``non-greedy'' part of our reduction.}
    \end{subfigure}
    \hfill
    \begin{subfigure}[c]{0.31\textwidth}
        \centering
        \includegraphics{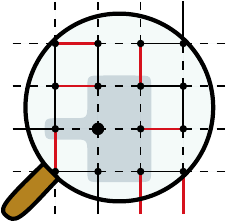}
        \caption{Finally, we show how each cluster can move those edges into its neighborhood while maintaining the matching property.}
    \end{subfigure}
    \caption{In order to complete the coloring on the edges within the clusters, our algorithm ensures that each cluster can independently control the colors of two edges adjacent to the cluster. For illustration purposes this example has many cycles that are dealt with separately in our algorithm.}
    \label{fig:matching-edges-clusters}
\end{figure}

To solve this non-greedy problem of assigning matching edges to clusters, we set up a hypergraph $H$.
Each cluster in $G$ corresponds to a vertex in $H$ and each edge $e \in M$ defines a hyperedge containing all clusters that sent at least one proposal to $e$.
Note that a hypergraph is in fact necessary here, since each edge in $M$ may receive proposals from up to $2\Delta$ different clusters. 
Now our assignment problem can then be viewed as a hypergraph orientation problem, where each hyperedge is said to be outgoing for exactly one of its vertices and ingoing for all others.
In this setting, the problem has already been studied in the literature under the name \emph{hypergraph sinkless orientation}.
As long as each cluster sends out sufficiently more proposals than any edge in $M$ receives, it can be solved in $\mathcal{O}(\log_\Delta n)$ rounds deterministically and $\mathcal{O}(\log_\Delta \log n)$ rounds randomized \cite{HSO}.
This is the only ``non-greedy'' part of our reduction and it dominates the runtime.
Every other step runs in $\mathcal{O}(\log^{\ast} n)$ rounds for constant-degree graphs, which shows that this is an essential part of our reduction in light of the $\Omega(\log_\Delta n)$ lower bound for edge-coloring with less than $2\Delta - 1$ colors.

Now as a final technical remark, while each cluster may own two assigned edges, these may not be directly adjacent to the cluster and controlling their color would still be useless in terms of obtaining solvability of the cluster.
To solve this issue we will show that cluster $C$ can move its assigned edges into its $1$-hop neighborhood along a shortest path to $C$, satisfying (1), while maintaining the matching property (2). 
Finally, the colorful condition is satisfied for every cluster and we can complete the $(2\Delta - 2)$-edge coloring for each cluster in parallel.

\subsection{Related work}
\label{subsec:related_work}

Distributed graph coloring is one of the most extensively studied problems in the field of distributed computing. Many foundational results are covered in the excellent book by Barenboim and Elkin \cite{barenboimelkin_book}. In this section, we survey both, results for greedy and non-greedy  edge coloring, highlighting a range of deterministic and randomized algorithms.

\subsubsection*{Greedy distributed edge coloring}

\subparagraph*{Deterministic algorithms}

In a seminal paper from 1987, Linial showed that any deterministic $3$-coloring algorithm on the $n$-cycle takes at least $1/2 \cdot \log^{\ast} n$ rounds in the \LOCAL model \cite{Linial1}. 
Naor later extended this lower bound to apply to randomized algorithms as well \cite{RandomizedLowerBound}.
On the upper bound side, Linial also presented an $\mathcal{O}(\log^{\ast} n)$-round $\mathcal{O}(\Delta^2)$-coloring algorithm \cite{Linial1}.
Combined with the classical one-color-per-round reduction this directly yields a $(2\Delta - 1)$-edge coloring algorithm in $\mathcal{O}(\Delta^2 + \log^{\ast} n)$ rounds.
With some clever tricks, subsequent works \cite{VertexColoringSublinear, LocallyIterative, VertexColoringLinear, SimpleAlgorithms, OnTheComplexity, LocalConflictColoring, LocalConflictColoringRevisited} have improved on this simple color reduction scheme, culminating in the current state-of-the-art for $(2\Delta - 1)$-edge coloring using $\mathcal{O}(\log^{12} \Delta + \log^{\ast} n)$ rounds \cite{EdgeColoringPolylogarithmicInDelta}.
Another line of research focuses on the complexity of graph coloring as a function of $n$, independent of $\Delta$.
Early works by Panconesi \& Srinivasan \cite{FirstDeltaColoringSTOC, NetworkDecompositionJournal} show that using network decompositions, one can find a $(2\Delta - 1)$-edge coloring in $\exp(\mathcal{O}(\sqrt{ \log n }))$ rounds. 
A long standing open question in the field asked for $(2\Delta-1)$-edge coloring algorithm in polylogarithmic time.
Barenboim and Elkin made progress by showing that the $\Delta^{1+\varepsilon}$-coloring problem can be solved in $\mathcal{O}(\log^2 n)$ rounds \cite{BoundedNeighborhoodIndependence}. 
Ghaffari \& Su later gave a $\mathrm{poly}(\log n)$ algorithm for $(2\Delta - 1)(1 + o(1))$-edge coloring \cite{DegreeSplitting}, which in turn was improved by \cite{DegreeSplittingJournal}.
Fischer, Ghaffari, and Kuhn \cite{EdgeColorHyperMaxMatch} fully resolved the question with a breakthrough $\mathcal{O}(\log^9 n)$-round algorithm based on a novel hypergraph maximal matching procedure.
Since then, the field has picked up the pace with a lot of new algorithms being developed in the last decade, e.g. \cite{EdgeColoringPolylogarithmicInDelta, Rounding, FasterMIS, GreedyAlgorithm, Derandomizing, ColoringWithoutNetDComp, 3Delta/2HMM, HarrisJournal}.
The fastest algorithm for general graphs currently comes from the MIS algorithm in \cite{GreedyAlgorithm} and runs in $\widetilde{\mathcal{O}}(\log^{5/3} n)$ rounds.

\subparagraph*{Randomized algorithms}

The first randomized algorithm for greedy edge coloring follows from the seminal $\mathcal{O}(\log n)$-round MIS algorithm independently developed by Luby \cite{LubySTOC} and Alon, Babai and Itai \cite{ABI}. This upper bound stood for almost three decades until it was surpassed by the $\mathcal{O}(\log \Delta) + 2^{\mathcal{O}(\sqrt{ \log \log n })}$-round algorithm introduced in \cite{LocalitySymmetryBreaking}, pioneering the now ubiquitous shattering framework for randomized algorithms. This algorithm was later improved to $\mathcal{O}(\log^* \Delta) + T^{\mathrm{det}}_{\mathrm{deg+1}}(\poly\log n)$ by Chang, Li, and Pettie \cite{OptimalColoring}. Here, $T^{\mathrm{det}}_{\mathrm{deg+1}}(n')$ is the deterministic complexity of the $(\mathrm{deg}+1)$-list coloring problem on instances of size $n'$ \cite{CLP20}. 
A major breakthrough came with the first polylogarithmic network decomposition algorithm \cite{NDCompBreakthrough}, reducing the complexity of the second term in both results and enabling $(2\Delta-1)$-edge coloring in $\mathrm{poly}(\log\log n)$ rounds.
Subsequent works \cite{Rounding,FasterMIS,GreedyAlgorithm,ImprovedNDComp,LowDiameterNDComp} continued to refine the runtime with the current state-of-the-art standing at $\widetilde{\mathcal{O}}(\log^{5/3} \log n)$ rounds \cite{GreedyAlgorithm}.

\subsubsection*{Non-greedy distributed edge coloring}

\begin{table}[H]
    \label{tab:edge-coloring}
    \centering
    \begin{tabular}{lcccc}
        \toprule
        \textbf{Citation} & \textbf{Year} & \textbf{Runtime} & \textbf{\# Colors} & \textbf{Random?} \\
        \midrule
        \cite{FirstDeltaColoringSTOC} & 1992 & $2^{\mathcal{O}(\sqrt{ \log n })}$ & $2\Delta - 2$ & \\
        \cite{3Delta/2} & 2018 & $\mathrm{poly}(\Delta, \log n)$ & $\Delta + \mathrm{poly}(\log n)$  & \\
        \cite{3Delta/2} & 2018 & $\mathcal{O}(\Delta^{9} \cdot \mathrm{poly} \log n)$ & $3\Delta/2$  & \\
        \cite{3Delta/2HMM} & 2019 & $\widetilde{\mathcal{O}}(\Delta^4 \cdot \log^6 n)$ & $3\Delta/2$  & \\
        \cite{VizingLocality} & 2019 & $\mathrm{poly}(\Delta, \log n)$ & $\Delta + \mathcal{O}(\log n / \log \log n)$  &  \\
        \cite{Bernshteyn, VizingBoundedDegree} & 2022 & $\widetilde{\mathcal{O}}(\Delta^{84} \cdot \log^5 n)$ & $\Delta + 1$ & \\
        \cite{HSO} & 2025 & $\mathcal{O}(\Delta^2 \cdot \log n)$ & $3\Delta/2$ & \\
        \cite{HSO} & 2025 & $\widetilde{\mathcal{O}}(\varepsilon^{-2} \cdot \log^2 \Delta \cdot \log n)$ & $(3/2 + \varepsilon)\Delta$ & \\
        \cite{FasterDeltaColoring} & 2025 & $\mathcal{O}(\log^4 \Delta + \log^2 \Delta \cdot \log n \log^\ast n)$ & $2\Delta - 2$ & \\
        \bottomrule \\
        \cite{NibbleMethod} & 1998 & $\mathcal{O}(\log n)$ & $(1+\varepsilon)\Delta$ & \dice{6} \\
        \cite{MuchEasier} & 2015 & $\mathcal{O}(\log^{\ast} \Delta \cdot \lceil \frac{\log n}{\Delta^{1 - o(1)}}\rceil)$ & $(1+\varepsilon)\Delta$ & \dice{6} \\
        \cite{DeltaColoring,NDCompBreakthrough} & 2018 & $\mathcal{O}(\log \Delta) + \poly(\log \log n)$ & $2\Delta - 2$ & \dice{6} \\
        \cite{SmallPalettes} & 2018 & $\mathrm{poly}(\log n)$ & $\Delta + \widetilde{\Theta}\left(\sqrt{ \Delta }\right)$ & \dice{6} \\
        \cite{VizingLocality} & 2019 & $\mathrm{poly}(\Delta, \log n)$ & $\Delta + 2$  & \dice{6} \\
        \cite{VertexSplitting} & 2022 & $\mathrm{poly}(\log\log n)$ & $(1+\varepsilon)\Delta$ & \dice{6} \\
        \cite{Bernshteyn, VizingBoundedDegree} & 2023 & $\mathcal{O}(\mathrm{poly}(\Delta) \cdot \log^2 n)$ & $\Delta + 1$ & \dice{6} \\
        \cite{Davies} & 2023 & $\mathrm{poly}(\varepsilon^{-1} \cdot \log \log n)$ & $(1+\varepsilon)\Delta$ & \dice{6} \\
        \cite{FasterDeltaColoring} & 2025 & $\widetilde{\mathcal{O}}(\log^{8/3}\log n)$ & $2\Delta - 2$ & \dice{6} \\
        \bottomrule
    \end{tabular}
    \caption{Edge-coloring algorithms in the LOCAL model using fewer than $2\Delta - 1$ colors.}
    \label{tab:edge_coloring}
\end{table}

Despite the more challenging nature of the problem, there has also been a lot of progress in the non-greedy regime of coloring with $k < 2\Delta - 1$ colors.
On the lower bound side, the authors of \cite{SmallPalettes} established that coloring the edges of a graph with at most $2\Delta - 2$ colors takes at least $\Omega(\log_\Delta n)$ rounds deterministically and $\Omega(\log_\Delta \log n)$ rounds randomized.
While most of the algorithms breaking the greedy color barrier were only developed in the last decade, there are two remarkable results before the turn of the millennium. 

\subparagraph*{$\Delta$-vertex coloring algorithms}

In 1992, Panconesi \& Srinivasan \cite{FirstDeltaColoringSTOC,FirstDeltaColoringJournal} presented the first distributed $\Delta$-coloring algorithm running in $\exp(\mathcal{O}(\sqrt{ \log n }))$ rounds. This deterministic result also implies the first distributed $(2\Delta - 2)$-edge coloring algorithm.
After two decades without improvements, the authors of \cite{DeltaColoring} finally decreased the upper bound for the randomized complexity of $\Delta$-coloring to $\mathcal{O}(\log \Delta) + 2^{\mathcal{O}(\sqrt{ \log \log n })}$.
Combining these seminal results with the new polylogarithmic algorithms for network decomposition \cite{NDCompBreakthrough,ImprovedNDComp,FasterMIS,Rounding,LowDiameterNDComp,GreedyAlgorithm} yields runtimes of $\mathrm{poly}(\log n)$ and $\mathcal{O}(\log \Delta) + \mathrm{poly}(\log \log n)$, respectively.
Recently, Bourreau, Brandt and Nolin \cite{FasterDeltaColoring} showed that $\Delta$-coloring can be solved in $\mathcal{O}(\log^4 \Delta + \log^2 \Delta \cdot \log n \log^\ast n)$ deterministic time and $\widetilde{\mathcal{O}}(\log^{8/3} \log n)$ randomized time. 
The randomized version of this result is also the current state-of-the-art for $(2\Delta - 2)$-edge coloring. 
Moreover, for $\Delta = \omega(\log^{21} n)$ the randomized complexity of $\Delta$-coloring collapses to $\mathcal{O}(\log^{\ast} n)$ due to \cite{DistributedBrooksTheorem}.
Finally, a very recent result \cite{NewDeltaColoring} gives an optimal $\Delta$-coloring algorithm for a restricted graph class.
However, since line graphs are not included in this class, it does not yield any new edge coloring results.

\subparagraph*{Randomized algorithms via the distributed Lovász Local Lemma}

In 1998, Dubhashi, Grable \& Panconesi \cite{NibbleMethod} presented an $(1+\varepsilon)\Delta$-edge coloring algorithm using the \emph{Rödl nibble method}, which runs in $\mathcal{O}(\log n)$ rounds, assuming $\Delta > (\log n)^{1+\gamma}$.
After the turn of the millennium, there has been no progress on the problem, until Elkin, Pettie \& Su \cite{MuchEasier} improved upon \cite{NibbleMethod} in 2015. 
Their randomized $(1+\varepsilon)\Delta$-edge coloring algorithm reduces the problem to the distributed Lovász Local Lemma and runs in $\mathcal{O}(\log^{\ast} \Delta +  \frac{\log n}{\Delta^{1 - o(1)}})$ rounds. 
In both cases, $\varepsilon$ need not be constant, but it is not clear how small it can be made as a function of $\Delta$. 
Focusing on palette size over speed, the authors of \cite{SmallPalettes} give a randomized $(\Delta + \widetilde{\Theta}(\sqrt{ \Delta }))$-edge coloring algorithm that runs in $\mathrm{poly}(\log n)$ rounds---nearly matching the natural threshold $\Delta + \Theta(\sqrt{ \Delta })$, the smallest palette size for which there still is a constant probability of being able to color any edge given a random feasible coloring of its neighborhood.
Recently, new advances in variants of the distributed Lovász Local Lemma \cite{VertexSplitting, Davies} have reduced the complexity of $(1+\varepsilon)\Delta$-edge coloring further to $\mathrm{poly}( \log \log n)$.

\subparagraph*{Deterministic algorithms}

In a 2018 paper \cite{3Delta/2}, the authors present a $\mathcal{O}(\Delta^9 \cdot \mathrm{poly}(\log n))$-time algorithm for $(3\Delta/2)$-edge coloring based on maximum matchings in bipartite graphs. 
Their result has since seen a series of improvements. 
First, Harris \cite{HarrisJournal} presented improved algorithms for hypergraph maximal matchings, reducing the runtime to $\widetilde{\mathcal{O}}(\Delta^4 \cdot \log^6 n)$ rounds. 
Next, a novel algorithm for hypergraph sinkless orientation \cite{HSO} further reduced the complexity to $\mathcal{O}(\Delta^2 \cdot \log n)$ rounds.
Using a splitting technique from \cite{DegreeSplittingJournal, DegreeSplittingConference} this algorithm can also be adapted to a $(3/2 + \varepsilon)\Delta$-edge coloring algorithm with runtime $\widetilde{\mathcal{O}}(\varepsilon^{-2} \cdot \log^2 \Delta \cdot \log n)$. 
Notably, this implies a $\mathcal{O}(\log^3 n)$-round algorithm for $(2\Delta - 2)$-edge coloring in general graphs.

\subparagraph*{Vizing chain algorithms}

A Vizing chain is an alternating path used to swap colors and resolve conflicts when trying to assign a new color to an edge in a graph that cannot be colored greedily. Vizing chains form a central role in Vizing's celebrated result stating that every graph is $(\Delta + 1)$-edge colorable \cite{Vizing}.
Recent years have brought significant progress on distributed versions of Vizing's theorem.
In \cite{VizingLocality}, the authors give the first $\mathrm{poly}(\Delta, \log n)$ algorithms for randomized $(\Delta + 2)$-edge coloring and deterministic $\Delta + \mathcal{O}(\log n / \log \log n)$ edge coloring.
Further, Bernshteyn and Dhawan \cite{Bernshteyn, VizingBoundedDegree} recently gave a deterministic $(\Delta + 1)$-edge coloring algorithm in $\widetilde{\mathcal{O}}(\Delta^{84} \cdot \log^5 n)$ rounds, together with a faster randomized version running in $\mathcal{O}(\mathrm{poly}(\Delta) \cdot \log^2 n)$ rounds.
While the exponent in $\Delta$ in the runtime is certainly not yet optimized, there is a fundamental barrier to this approach.
The authors of \cite{SmallPalettes} have shown that it requires recoloring subgraphs of diameter $\Omega(\Delta \cdot \log n)$ in the worst case.

\subsection{Organization of the paper}

In \Cref{sec:technical} we present our main technical contribution: a sufficient condition under which partial $(2\Delta - 2)$-edge colorings can be extended.
Building on this, we introduce our new algorithms in \Cref{sec:alg}.
We start by presenting our reduction from $(2\Delta - 2)$-edge coloring to MIS. 
Later on, in \Cref{subsec:optimized}, we replace the MIS and maximal matching subroutines to achieve a reduction to just $(2\Delta - 1)$-edge coloring.
Then, in \Cref{sec:randomized} we adapt the reduction to obtain a faster randomized algorithm.
For convenience, we collect all subroutines from the literature that we use in our algorithms in the appendix, to serve as a point of reference.

\section{Preliminaries}

We will use the notation $[k] := \{ 1,\dots, k \}$.
Let $G = (V,E)$ be an undirected graph.
The \emph{neighborhood} of a vertex $v \in V$ in $G$ is $N_G(v) = \{ w \in V: vw \in E \}$ 
and the \emph{edge-neighborhood} of $v$ is $N_E(v) = \{ e \in E: \lvert e \cap \{ v \} \rvert = 1 \}$.
Further, for a subset of vertices $W \subseteq V$ we write $N_E(W) = \{ e \in E: \lvert e \cap W \rvert = 1\}$.
The \emph{distance} between two vertices $u,v \in V$ is defined as the number of edges on a shortest path connecting $u$ and $v$ in $G$ and is denoted by $\mathrm{dist}(u,v)$. We define the distance between an edge $e = uv \in E$ and a vertex $w \in V$ as 
$$
\mathrm{dist}(e,w) := \min(\mathrm{dist}(u,w),\mathrm{dist}(v,w)).
$$
For any positive integer $k > 0$ we define the \emph{power graph} $G^k = (V,E^k)$ by adding an edge for every pair of vertices $u, v \in V$ with $\mathrm{dist}(u,v) \leq k$ to $E^k$.
For any positive integer $k > 0$ and any vertex $v \in V$ we define the \emph{$k$-hop neighborhood} of $v$ to be $N^k(v) = N_{G^k}(v)$.
For a subset of edges $F \subseteq E$ we write $G - F$ to denote the graph obtained by removing all edges in $F$ from $G$.
Further, for a subset $V' \subseteq V$, we denote by $G[V'] = (V',E')$ the subgraph of $G$ \emph{induced} by $V'$. 
The edge set $E'$ is given by the subset of edges with both endpoints in $V'$. 
Similarly, for a subset of edges $E' \subseteq E$ we define the \emph{edge-induced subgraph} $G[E'] = (V',E')$,
where $V'$ is given by the subset of vertices that are incident to at least one edge in $E'$.
We define the \emph{degree} $\mathrm{deg}(e)$ of an edge $e \in E$ as the number of edges adjacent to $e$.
A \emph{hypergraph} is a generalization of a graph in which each edge can contain more than two vertices. The \emph{rank} of a hyperedge $e$ is the number of vertices contained in $e$.
A \emph{proper $c$-edge coloring} is a function $\varphi: E \to [c]$ that assigns different colors to adjacent edges.
In this paper we will work a lot with \emph{partial edge colorings} that assign colors only to a subset of edges.

\begin{definition}
    Let $G = (V,E)$ be a graph, $F \subseteq E$ and $\varphi: E \setminus F \to [2\Delta - 2]$ a partial edge coloring of $G$. Then we call $\varphi^\star$ an \emph{extension} of $\varphi$ to $F$ if $\varphi^\star$ is a proper $(2\Delta - 2)$-edge coloring of $G$ and $\varphi^\star(e) = \varphi(e)$ for all $e \in E \setminus F$. Further, we say a color $c$ is \emph{available} for an edge $e \in F$, if no adjacent edge already has color $c$.
\end{definition}

\begin{definition}[$(\alpha,\beta)$-ruling set]
    For two integers $\alpha, \beta \geq 1$, an \emph{$(\alpha,\beta)$-ruling set} for a subset $W \subseteq V$ is a set of nodes $S \subseteq V$ such that the distance between any two nodes in $S$ is at least $\alpha$ and any node $w \in W$ has a distance of at most $\beta$ to the closest node in $S$. If $V = W$, we call $S$ an $(\alpha,\beta)$-ruling set.
\end{definition}

Ruling sets generalize maximal independent sets, which are $(2,1)$-ruling sets. 
In our algorithm we will use ruling sets as a tool to partition the vertex sets into clusters with some specific properties. 
Most importantly, we will need that for each cluster $C$ there is a central vertex $r \in C$ such that the whole $k$-hop neighborhood around $r$ is contained in $C$, for some parameter $k > 0$. 
This ensures that each cluster has sufficiently many vertices, which will become important when we try to assign two exclusive matching edges to each cluster.

\begin{definition}[$(\alpha,\beta)$-clustering]
    Let $G = (V,E)$ be a graph and $\mathcal{R} \subseteq V$ a subset of vertices. Then, we call a partition $\bigcup_{r \in \mathcal{R}} C(r) = V$ of the vertex set into pairwise disjoint subsets a \emph{$(\alpha,\beta)$-clustering} with respect to $\mathcal{R}$ if $N^{\alpha}(r) \subseteq C(r)$ and $\mathrm{diam}(C(r)) \leq \beta$ for all $r \in \mathcal{R}$.
\end{definition}

\begin{remark}
    Computing a $(2,\beta)$-ruling set on the power graph $G^k$ and letting each vertex join the cluster of one of its closest ruling set nodes leads to a $(\lfloor \frac{k}{2} \rfloor , k \cdot \beta)$-clustering. 
\end{remark}

In the third step of our algorithm, our goal is to assign two exclusive edges to each cluster, such that the set of all such edges forms a matching in $G$. 
Starting with a maximal matching, we construct an auxiliary hypergraph $H$ that lets us modify the matching in order to satisfy this property.
This can be seen as an instance of the \emph{hypergraph sinkless orientation} problem, which is a natural generalization of the fundamental \emph{sinkless orientation} problem.

\begin{definition}[Hypergraph sinkless orientation (HSO)]
    Let $H = (V,E)$ be a hypergraph. The objective of the \emph{hypergraph sinkless orientation} problem is to orient the hyperedges of $H$ such that every vertex $v \in V$ has at least one outgoing hyperedge. We define an oriented hyperedge to be outgoing for exactly one of its incident vertices and incoming for all others.
\end{definition}

If the minimum degree $\delta = \min_{v \in V} \mathrm{deg}(v)$ and the maximum rank $r = \max_{e \in E} \mathrm{rank}(e)$ of $H$ satisfy $\delta > r$, then this problem can be solved efficiently using the algorithms in \Cref{thm:hso,thm:hso_randomized}.
We follow the standard assumption that each vertex in the \LOCAL model knows the number of nodes $n = \lvert V \rvert $ of the graph, as well as the maximum degree $\Delta$.

\section{Under which condition can we extend partial \texorpdfstring{$(2\Delta - 2)$}{(2Δ-2)}-colorings?}
\label{sec:technical}

The goal of this section is to establish conditions under which we can extend $\varphi$ to a proper $(2\Delta - 2)$-edge coloring of the entire graph $G$.
As a first step, we examine the case where the subgraph induced by the uncolored edges is a simple star.
In this setting, the only forbidden configuration occurs when all the leaves of the star graph share the same set of $\Delta - 1$ incident colors, which can be seen in \Cref{fig:forbidden_configuration}.

\begin{lemma}[Reaching for the stars]
    \label{lem:stars}
    Let $G = (V,E)$ be a graph and let $\varphi$ be a partial $(2\Delta-2)$-edge coloring of $G$
    that only leaves the edges of a star graph $T = (V_T,E_T) \subseteq G$ uncolored. 
    If $\lvert \varphi(N_E(V_T)) \rvert \geq \Delta$, then there is an extension $\varphi^{\star}$ of $\varphi$ to $E_T$.
\end{lemma}

To prove this lemma we invoke the following classical combinatorial result.

\begin{theorem}[Hall's theorem \cite{Hall}]
    \label{thm:hall}
    Let $G = (V \,\dot\cup\, U, E)$ be a bipartite graph. For any subset $S \subseteq U$, let $N(S)$ denote the neighborhood of $S$ in $G$. Then $G$ has a $U$-saturating matching, if and only if $|N(S)| \geq |S|$ for all $S \subseteq U$.
\end{theorem}

\begin{proof}[Proof of \Cref{lem:stars}]
    We construct the auxiliary bipartite graph $B = (E_T \cup [2\Delta-2], F)$ as follows: 
    For every edge $e \in E_T$ and every color $c \in [2\Delta -2]$ we add the edge $(e,c)$ to $F$ if and only if $c$ is available for $e$ under $\varphi$. 
    Observe that any $E_T$-saturating matching in $B$ corresponds to an extension of $\varphi$ to $E_T$.
    By \Cref{thm:hall} such a matching exists if and only if for any subset $S \subseteq E_T$ it holds that $\lvert N_B(S) \rvert \geq \lvert S \rvert$.
    Since any edge in $T$ has at least $\Delta - 1$ colors available and $\lvert E_T \rvert  \leq \Delta$, the only set left to check is $S = E_T$.
    For contradiction sake, assume that $\lvert N_B(E_T) \rvert = \Delta - 1$ and let $r \in V_T$ be the center of the star. 
    This implies that for every $e = rv \in E_T$ all colors in $C := [2\Delta - 2] \setminus N_B(E_T)$ are already occupied. 
    Since $\lvert C \rvert = \Delta - 1 \geq \lvert N_E(v) \cap N_E(V_T)\rvert $, this is only possible if $\varphi(N_E(v) \cap N_E(V_T)) = C$ for all $e \in E_T$ and therefore $\varphi(N_E(V_T)) = C$.
    But $\lvert \varphi(N_E(V_T)) \rvert \geq \Delta$, a contradiction!
\end{proof}

As the next step we show that this condition generalizes to trees as well. 
More precisely, we will show that if $\varphi$ assigns at least $\Delta$ distinct colors to the edges adjacent to at least one vertex at distance $k$ from the root, then there is always an extension of $\varphi$ that assigns at least $\Delta$ distinct colors to the edges adjacent to any vertex at distance $k - 1$ from the root. 

\begin{lemma}[Colorful leaves make any tree happy]
    \label{lem:colorful_leaves}
    Let $G = (V,E)$ be a $\Delta$-regular graph, $T = (V_T,E_T) \subseteq G$ a tree, $r \in V_T$ and $\varphi: E \setminus E_T \to [2\Delta - 2]$ a partial edge coloring of $G$. Let 
    $$
    V_{k} := \{ v \in V_T: \mathrm{dist}_T(v,r) = k \} \qquad \text{and} \qquad E_k = \{ e \in E_T: \mathrm{dist}_T(e,r) = k \}.
    $$
    If there is a $k \in \mathbb{N}$ such that $|\varphi(N_E(V_k))| \geq \Delta$, then there exists an extension $\varphi^{\star}$ of $\varphi$ to $E_T$.
\end{lemma}

\begin{figure}
    \captionsetup{justification=centering}
    \centering
    \begin{subfigure}[t]{0.46 \textwidth}
        \centering
        \includegraphics{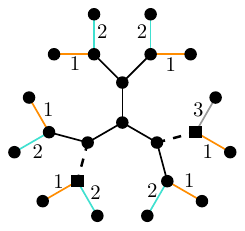}
        \caption{The edges with different colors palettes available are dashed.}
    \end{subfigure}
    \hfill
    \begin{subfigure}[t]{0.46 \textwidth}
        \centering
        \includegraphics{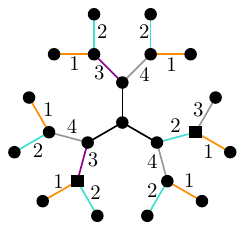}
        \caption{The coloring can be extended while using at least three colors for the newly colored edges.}
    \end{subfigure}
    \caption{Extending the coloring to a tree.}
    \label{fig:tree_graph}
\end{figure}

\begin{proof}
    Let $\ell \in \mathbb{N}$ such that $|\varphi(N_E(V_{\ell}))| \geq \Delta$.
    First we observe that for every $k > 0$, each edge $e \in E_k$ has exactly one neighboring edge in $E_{k-1}$.
    Therefore, we can always find at least one available color for every such edge, as long as the set $E_{k-1}$ remains uncolored. 
    Using this observation we can assign arbitrary available colors to all edges in $\bigcup_{k\geq \ell} E_{k}$ by processing these edge sets in sequence, starting with the set with highest index and going through each of these sets in an arbitrary order.
    However, after we finish coloring $E_{\ell}$ we need to start being more deliberate with our color choices. 
    For any $k < \ell$ we need to ensure that $\lvert \varphi(N_E(V_k)) \rvert \geq \Delta$ holds after we colored all edges of $E_{k}$. This is due to the fact that we need $\lvert \varphi(N_E(V_1))\rvert \geq \Delta$ in order to apply \Cref{lem:stars}, which guarantees that we can complete the coloring on the final layer of edges $E_0$ connected to the root. 
    Now we show that given $\lvert \varphi(N_E(V_{k+1})) \rvert  \geq \Delta$ we can color the edges in $E_k$ in such a way that $\lvert \varphi(N_E(V_{k})) \rvert \geq \Delta$ for any $1 \leq k < \ell$.
    For each $e \in E_k$ let $N^+(e) := N_E(e) \cap N_E(V_{k+1})$. Since $\lvert N^+(e) \rvert = \Delta - 1$ we can always find two edges $e,e' \in E_{k}$ such that $\varphi(N^+(e)) \neq \varphi(N^+(e'))$. 
    We first assign colors to all edges in $E_{k}$ that are not adjacent to $e$ or $e'$. 
    Since we assume $\Delta$ to be sufficiently large, there always exists a vertex $v \in V_k$ that is not incident to either $e$ or $e'$. 
    Thus, after coloring all $\Delta - 1$ edges in $N_E(v) \setminus E_{k-1}$ incident at $v$ we already have $\lvert \varphi(N_E(V_{k})) \rvert \geq \Delta - 1$. 
    If $\lvert \varphi(N_E(V_{k})) \rvert \geq \Delta$, we are already done. 
    Otherwise, if $\lvert \varphi(N_E(V_{k})) \rvert = \Delta - 1$, then the set $\varphi(N_E(V_{k}) \cup (N^+(e) \cap N^+(e')))$ contains at most $2\Delta - 3$ colors. 
    Therefore, either $e$ or $e'$ must have a color available in $[2\Delta - 2] \setminus \varphi(N_E(V_{k}))$.
    We pick that color for $e$ or $e'$, which ensures that $\lvert \varphi(N_E(V_{k})) \rvert \geq \Delta$. 
    The remaining edges in $E_k$ can then be assigned arbitrary available colors again. 
    For an illustration of this procedure we refer to \Cref{fig:tree_graph}.
    Finally, to color the last layer $E_{0}$ containing the edges incident to the root, we can apply \Cref{lem:stars}, since the number of distinct adjacent colors is at least $\Delta$.
\end{proof}

Now in order to fulfill this condition, our algorithm assigns two exclusive edges to each cluster. 
Each cluster may then decide independently, whether it wants to keep the current color or change it to the (yet unused) color $2\Delta - 2$ for each assigned edge. 
Hence, the set of all those edges needs to be independent. 
Therefore we initially compute a maximal matching on $G - E_T$ and modify it in a way such that each cluster receives at least two edges.
We model this adjustment of $M$ as an instance of hypergraph sinkless orientation, where each vertex proposes to a suitable matching edge in its $2$-hop neighborhood. 
The core requirement for this problem is that the minimum degree on the vertex side must exceed the maximum rank on the hyperedge side. 
In other words, each cluster must send more proposals to distinct edges than the number of proposals that any edge receives. 
The details of this procedure will be explained in \Cref{subsec:HSO}.
Therefore we need that the cluster expand quickly. 
This is only the case if the are not a lot of cycles present.
Fortunately, this case is easy, since as soon as a connected component contains one cycle of even length or a vertex of degree less than $\Delta$, we can always extend any partial $(2\Delta - 2)$ edge coloring to the edges of this connected component. 

\begin{restatable}{lemma}{degreechoosable}
    Let $G = (V,E)$ and $\varphi: E \to [2\Delta -2] \cup \{ \bot \}$ be a partial $(2\Delta-2)$-edge coloring of $G$.
    Assume that the subgraph $H \subseteq G$ induced by the uncolored edges is connected and contains either a cycle of even length or a vertex of degree less than $\Delta$.
    Then there exists an extension $\varphi^{\star}$ of $\varphi$, which is a proper $(2\Delta-2)$-edge coloring of the entire graph $G$.
\end{restatable}

\begin{proof}
    Let $C$ be a cycle of even length and assume that all edges in $N(C)$ have been colored already.
    Our task is now to find a proper edge-coloring of $C$, where every edge has a list of available colors of size two.
    If all edges have the same two colors available, then we can color the cycle alternately with these two colors.
    Otherwise, there exist two adjacent edges $e_{1}, e_{2}$ with different sets of available colors. 
    Color $e_{1}$ with a color that is not available for $e_{2}$. 
    Then we traverse the cycle starting at $e_{1}$ and going in the opposite direction of $e_{2}$ and let each edge pick one of its available colors. 
    Since $e_{1}$ picked a color that is not available at $e_{2}$, there is also a color remaining for $e_{2}$ at the end.
    Since $H$ is connected, every edge lies on a path towards an even cycle. 
    Fix such an even cycle $C$ and define
    $$
        E_{k} = \{ e \in E(H): \mathrm{dist}(e,C) = k \}.
    $$
    Since every edge in $E_{k}$ has at least one neighbor in $E_{k-1}$ for $k > 1$, we can greedily assign colors to these edges, starting from the highest layer. 
    Finally, every edge in $C$ has a list of two available colors and can therefore be properly colored as well.
\end{proof}

On the other hand, if a cluster does not contain a cycle of even length, we need to ensure that it can satisfy the colorful condition from \Cref{lem:colorful_leaves}. 
To this end, it is crucial that the cluster contains a sufficient number of leaves.
We apply a similar argument to the one used in the analysis of the $\Delta$-vertex coloring algorithm in \cite{DeltaColoring} to show that a connected component that does not contain an even cycle must be expanding exponentially.

\begin{restatable}{lemma}{clusterexpansion}
    \label{lem:cluster_expansion}
    Let $G = (V,E)$ be a $\Delta$-regular graph, $k > 0, v \in V$ and $C \subseteq G$ be a connected subgraph of $G$ such that $N^k(v) \subseteq C$.
    Further assume that $C$ does not contain an cycle of even length. 
    Then, any BFS-tree $T$ of $C$ rooted at $v$ contains at least $(\Delta - 2)^{k}$ leaves.
\end{restatable}

\begin{proof}
    For two nodes $u, u' \in V(C)$ let $P(u,u')$ denote the unique path in $T$ between $u$ and $u'$.
    We claim that each node $u \in V(C)$ is incident to at most one edge that is not part of $E(T)$ and thus the degree of all internal nodes of $T$ is at least $\Delta - 1$.
    First we observe that any edge in $E(C) \setminus E(T)$ can only connect vertices which are at most one level apart in $T$.
    Any additional edge in $E(C) \setminus E(T)$ from a vertex $v_{\ell}$ in layer $\ell$ to a vertex $v_{\ell+1}$ in layer $\ell + 1$ in $T$ immediately induces an even cycle. 
    To wit, let $u$ be the least common ancestor of $v_{\ell}$ and $v_{\ell+1}$ and observe that the cycle $P(u,v_{\ell}) \cup \{ v_{\ell}, v_{\ell+1} \} \cup P(u,v_{\ell+1})$ has length $2 \cdot \lvert P(u,v_{\ell+1}) \rvert$.
    Next, let $u \in V(C)$ be a vertex in layer $\ell$ with two neighbors $w,w'$ in the same layer $\ell$. Then, $P(v,w) \cup \{ w,u \} \cup \{ u, w' \} \cup P(v,w')$ is a cycle of length $2(\ell+1)$. 
    Hence each vertex $u \in V(C)$ is adjacent to at most one edge from $E(C) \setminus E(T)$.
    Since $C$ contains the complete $k$-hop neighborhood of $v$ this implies that there are at least $(\Delta - 2)^k$ vertices in layer $k$ of $T$. 
    Every such vertex either has a leaf-descendant in $T$ or is a leaf itself.
\end{proof}

\section{Deterministic \texorpdfstring{$(2\Delta - 2)$}{(2Δ-2)}-edge coloring algorithm}
\label{sec:alg}

We first present the simpler version of our algorithm, which reduces the $(2\Delta - 2)$-edge coloring problem deterministically to MIS in $\mathcal{O}(\log_\Delta n)$ rounds of the \LOCAL model. 

\subsection{High level overview of our reduction to MIS}
\label{subsec:simple_alg}
Recall that HSO (\Cref{thm:hso}) aims to orient the hyperedges of a hypergraph such that each vertex has at least one outgoing hyperedge. 
In Phase 2 of \Cref{alg:coloring}, it is used to rearrange the existing matching, ensuring that each cluster is assigned two exclusive edges.

\begin{algorithm}[ht!]
	\caption{Reduction from $(2\Delta-2)$-edge coloring to MIS (high level overview)}
    \label{alg:coloring}
	\begin{algorithmic}[1]
	\Statex \textbf{Input:} a graph $G = (V,E)$
    \Statex \textbf{Output:} $(2\Delta-2)$-edge coloring $\varphi$ of $G$.
    \Statex \hspace{-0.5cm}\textbf{Phase 1: Partition vertices into clusters (\Cref{subsec:clustering})}
    \State Compute an MIS $\mathcal{I} \subseteq V$ on the power graph $G^8$
    \Comment \Cref{alg:greedy}
    \State $\mathcal{C}, E_T \gets$ \textsc{Cluster}$(G, \mathcal{I})$
    \Comment \Cref{alg:clustering}
    \State Compute a $(2\Delta - 3)$-edge coloring $\varphi$ of $G - E_T$
    \Comment{\Cref{lem:2delta-3}} 
    \Statex \hspace{-0.5cm}\textbf{Phase 2: Assign two exclusive edges to each cluster (\Cref{subsec:HSO})}
    \State Compute a maximal matching $M$ of $G - E_T$
    \Comment{\Cref{lem:max_match}}
    \State Set up auxiliary hypergraph $H = (\mathcal{C}, M)$
    \Comment \Cref{def:hypergraph}
    \State Assign two matching edges to each cluster via HSO on $H$
    \Comment \Cref{thm:hso}
    \State Move the assigned edges into the neighborhood of the respective clusters
    \Comment \Cref{lem:proper_matching}
    \Statex \hspace{-0.5cm}\textbf{Phase 3: Switch colors in order to complete the coloring (\Cref{subsec:color_switch})}
    \For{each cluster in parallel}
        \State Change the colors of the assigned edges to fulfil the colorful condition
        \Comment{\Cref{lem:color_switching}}
    \EndFor
    \State Extend $\varphi$ to a $(2\Delta - 2)$-edge coloring $\varphi^\star$ of $G$
    \Comment \Cref{lem:colorful_leaves}
	\end{algorithmic}
\end{algorithm}

\subparagraph*{Phase 1: Partition vertices into clusters (\Cref{subsec:clustering}).} 

We start by computing a maximal independent set $\mathcal{I}$ on the power graph $G^8$. 
Next we partition the vertex set into disjoint clusters.
Each vertex joins the cluster of one of its closest MIS-nodes, breaking ties arbitrarily, but consistently, which leads to a $(4,8)$-clustering $\mathcal{C}$. 
In addition, we store the BFS-tree $(V,E_T)$ which arises naturally from this clustering.
Computing the MIS on $G^8$ ensures that the individual clusters all contain sufficiently many vertices, which will become important later.
Next, we remove the union of all BFS-trees from $G$ to obtain a subgraph $G - E_T \subseteq G$ of maximum degree at most $\Delta - 1$. 
Hence, it suffices to invoke a greedy $(2\Delta - 1)$-edge coloring on $G - E_T$ in order to get a $(2\Delta  - 3)$-edge coloring $\varphi$ of $G - E_T$.

\subparagraph*{Phase 2: Assign two exclusive edges to each cluster (\Cref{subsec:HSO}).}

Finally we need to modify $\varphi$ such that the colorful condition from \Cref{lem:colorful_leaves} is satisfied for every cluster. 
We will show that this can be achieved by giving each cluster control over the colors of two edges adjacent to the cluster. 
However, in order for the clusters to switch the colors of their respective assigned edges, the union of all those edge sets needs to be independent.
Therefore we initially compute a maximal matching $M$ of $G - E_T$ to use as a starting point. 
Note that the maximality of the matching does not guarantee that any fixed cluster receives any adjacent matching edges, but only that there are sufficiently many matching edges in the $2$-hop neighborhood around that cluster.
Hence, we then let each cluster propose to matching edges in its $2$-hop neighborhood.
This process naturally defines a hypergraph, where each matching edges collects the clusters which sent a proposal to it as vertices in a hyperedge.
Then, the problem of assigning matching edges to clusters can be modeled as an instance of hypergraph sinkless orientation.
In order to be able to solve HSO efficiently we need to ensure that the number of votes each matching edge receives is dominated by the number of votes each cluster casts in total. 
Here we require that each cluster contains enough vertices.

\subparagraph*{Phase 3: Switch colors in order to complete the coloring (\Cref{subsec:color_switch}).}

Finally, we show that each cluster can flip the colors of their assigned matching edges in parallel in order to fulfil the colorful condition from \Cref{lem:colorful_leaves}. 
Thus, this new coloring can now be extended to the remaining uncolored edges by iteratively coloring the layers of each BFS-tree in parallel. 
Since each cluster has diameter of at most $8$ this last step can be executed in constant time.

\MISalgo*

\subsection{Phase 1: Partition vertices into clusters}
\label{subsec:clustering}

\begin{algorithm}
	\caption{\textsc{Cluster}$(G,\mathcal{I})$}
    \label{alg:clustering}
	\begin{algorithmic}[1]
	\Statex \textbf{Input:} a graph $G = (V,E)$ and an independent set $\mathcal{I} \subseteq V$
    \Statex \textbf{Output:} a $(4,8)$-clustering $\mathcal{C} = \{ C^{-1}(r): r \in \mathcal{I} \} $ and a forest $T = (V,E_T)$
    \State $E_T \gets \emptyset$
	\For{each $v \in V$ in parallel} 
		\State $C(v) \gets \bot$
		\If{$v \in \mathcal{I}$}
			\State $C(v) \gets v$
	    \EndIf
	    \While{$C(v) = \bot$}
            \State receive messages from neighbors
		    \If{$\exists\, u \in N(v): C(u) \neq \bot$}
			    \State $C(v) \gets C(u)$ 
                \State $E_T \gets E_T \cup \{ uv \} $
                \Comment{break ties arbitrarily}
                \State send message $C(v)$ to neighbors
            \EndIf
        \EndWhile
    \EndFor
	\end{algorithmic}
\end{algorithm}

\begin{lemma}[Clustering algorithm]
    The set $\mathcal{C} = \{ C^{-1}(r): r \in \mathcal{I} \} $ computed by \Cref{alg:clustering} is a $(4,8)$-clustering with respect to $\mathcal{I}$.
\end{lemma}

\begin{proof}
    First we show that $N^4(r) \subseteq C^{-1}(r)$ for all $r \in \mathcal{I}$. 
    Let $r \in \mathcal{I}, v \in N^4(r)$ and assume for contradiction sake that $C(v) = r' \neq r$. 
    Then, $\mathrm{dist}(v,r') \leq \mathrm{dist}(v,r) \leq 4$ and therefore $\mathrm{dist}(r,r') \leq 8$, contradicting the fact that $r$ and $r'$ are independent in $G^8$.
    For the second property assume that there exist $r \in \mathcal{I}$ and $v \in C^{-1}(r)$ such that $\mathrm{dist}(v,r) > 8$. 
    This implies that $\mathrm{dist}(v,r') \geq \mathrm{dist}(v,r) > 8$ for all $r' \in \mathcal{I}$, because otherwise $v$ would have joined the cluster of $r'$ instead of $r$.
    But then $v$ would be independent of all $r' \in \mathcal{I}$ on $G^8$, which contradicts the maximality of $\mathcal{I}$.
\end{proof}
For each $r \in \mathcal{I}$, let $T(r) = T[C^{-1}(r)]$ denote the BFS-tree of the cluster $C^{-1}(r)$ rooted at $r$.

\begin{lemma}
    \label{lem:2delta-3}
    The subgraph $G - E_T$ has maximum degree of at most $\Delta - 1$ and can be $(2\Delta - 3)$-edge colored in $T_{\mathrm{MIS}}(\Delta^2 \cdot n,4\Delta)$ rounds of the \LOCAL model.
\end{lemma}

\begin{proof}
    We observe that the edge set $E_T$ returned by \Cref{alg:clustering} satisfies $E_T = \bigcup_{r \in \mathcal{I}} E(T(r))$.
    For every vertex $v \in V$ there is an $r \in \mathcal{R}$ such that $v \in C^{-1}(r)$. Since $T(r)$ is a spanning tree of $C^{-1}(r)$, we get that $v$ is incident to at least one edge in $E(T(r))$. Thus, the maximum degree of $G - E_T$ is at most $\Delta - 1$.
    Using the standard reduction from vertex coloring to MIS \cite{LubySTOC} we can compute a $(2\Delta - 3)$-edge coloring of $G - E_T$ in $T_{\mathrm{MIS}}(\Delta^2 \cdot n,4\Delta)$ rounds.
\end{proof}

\begin{lemma} \label{lem:max_match}
    We can compute a maximal matching $M$ of the edges in $G - E_T \subseteq G$ deterministically in $T_{\mathrm{MIS}}(\Delta \cdot n,2\Delta - 2)$ rounds of the \LOCAL model.
\end{lemma}

\begin{proof}
    A maximal matching of $G - E_T$ is simply an MIS on the line graph $L(G - E_T)$.
\end{proof}

\subsection{Phase 2: Assign two exclusive edges to each cluster}
\label{subsec:HSO}

In this phase, we rearrange the edges in the maximal matching $M$ in order to assign two exclusive edges to each cluster, while preserving the matching property. 
These exclusive edges are essential for adjusting the coloring of the intercluster edges, as laid out in \Cref{lem:color_switching}. 
This step relies on the HSO subroutine from \Cref{thm:hso}, and we now detail the specific instance used in our algorithm.
The high level idea is that each cluster sends proposals to nearby edges in the matching $M$. 
Each edge accepts exactly one proposal and we swap the edge towards the cluster that \emph{won} it. 
This process can be modeled as an HSO instance where each matching edge in $M$ represents an hyperedge that contains as vertices all clusters that sent a proposal to that edge. 
Since a directed hyperedge $e$ is outgoing for exactly one of its vertices $v \in e$, we will say that $v$ \emph{won} the proposal for $e$. 

However, the hypergraph sinkless orientation problem only guarantees at least one outgoing hyperedge for each vertex. 
Therefore, in order to achieve two exclusive hyperedges for each cluster, we simply split each cluster into two virtual nodes and let both participate in the HSO independently. 
In the following we formalize this idea and prove that this provides the desired properties as well as a deterministic runtime of $\mathcal{O}(\log n)$ rounds. 
First, we let each cluster $C^{-1}(r) \in \mathcal{C}$ that is not already edge-degree choosable select a subset of vertices $S \subseteq V(C^{-1}(r))$ to send proposals to a matching edges in their neighborhood. 
Let $\mathcal{C}'$ denote the set of clusters, which are not already edge-degree choosable. 
Crucially, in order to satisfy the colorful condition from \Cref{lem:colorful_leaves} we require that the vertices in $S$ are all leaves in the spanning tree $T(r) \subseteq C^{-1}(r)$ at the same distance from the root $r$.

\begin{lemma}
    \label{lem:enough_vertices}
    For each cluster $C^{-1}(r) \in \mathcal{C}'$, there exists an integer $k > 0$ and a subset $S \subseteq V(C^{-1}(r))$ of size at least $4 \cdot \Delta^3$ such that every vertex $v \in S$ is a leaf in the tree $T(r)$ that satisfies $\mathrm{dist}(r,v) = k$.
\end{lemma}

\begin{proof}
    Since $\mathcal{C}$ is a $(4,8)$-clustering we have that $C^{-1}(r)$ necessarily contains all vertices in $N^4(r)$. 
    Thus, according to \Cref{lem:cluster_expansion} the corresponding BFS-tree contains at least $(\Delta - 2)^4$ leaves. 
    Therefore, by the pigeonhole principle there exists a layer $V_k$ for $k = 1,\dots,8$ such that $V_k$ contains at least $(\Delta - 2)^4 / 8 > 4\cdot \Delta^3$ leaves for $\Delta$ sufficiently large.
\end{proof}

For each cluster we select an arbitrary set $S$ of at least $4 \cdot \Delta^3$ leaves which are all at the same distance from the center. 
The subset $S$ is the subset of leaves that send requests. 
Each vertex in $v \in S$ sends exactly one proposal. 
If $v$ is incident to a matching edge $e \in M$, then it proposes to $e$. 
Otherwise, $v$ sends a proposal to an arbitrary matching edge that is one hop away from $e$. Such an edge must exist, due to the maximality of $M$.

\begin{definition}[Auxiliary hypergraph] \label{def:hypergraph}
    We define the \emph{auxiliary hypergraph} $H = (V_H, E_H)$:
    \begin{itemize}
        \item For each cluster $C \in \mathcal{C}'$ we add a vertex $v_C$ to $V_H$.
        \item For each matching edge $e \in M$ we add a hyperedge $e_H$ to $E_H$ that contains all vertices $v_C$ for which there exists at least one vertex $v \in V(C)$ that proposed to $e$.
    \end{itemize}
\end{definition}

\begin{lemma} \label{lem:hsorankdeg}
    The minimum degree $\delta_H$ and the maximum rank $r_H$ of $H$ satisfy $\delta_H > \Delta \cdot r_H$.
\end{lemma}

\begin{proof}
    First we observe that a single matching edge $e = uv$ can get proposals only from $u, v$ and vertices that are exactly one hop away from either $u$ or $v$. 
    Hence, there may be two proposals coming from $u$ and $v$ plus $2 \cdot (\Delta - 1)$ proposals coming from neighbors of $u$ or $v$. 
    Thus, the maximum rank of $H$ is at most $r_H = 2\Delta$.
    Next, each cluster contains at least $4 \cdot \Delta^3$ vertices sending proposals. 
    Since each vertex sends exactly one proposal and each edge receives at most $2\Delta$ proposals the vertices of each cluster must propose to at least $2\Delta^2$ edges in total. 
    Hence, the minimum degree of $H$ can be lower bounded by $\delta = 2\Delta^2$.
\end{proof}

\begin{lemma} \label{lem:runtime_hso} 
    We can compute an orientation of $E_H$ such that each vertex in $V_H$ has at least two outgoing hyperedges in $\mathcal{O}(\log_\Delta n)$ rounds deterministically and $\mathcal{O}(\log_\Delta \log n)$ rounds randomized.
\end{lemma}

\begin{proof}
    To compute an orientation with two outgoing edges per cluster we simply split each vertex into two virtual nodes with half the degree each. 
    Thus, the minimum degree of the resulting hypergraph $H'$ is still at least $\Delta^2$. 
    Therefore, we can solve HSO in $\mathcal{O}(\log_{\delta / r} n) = \mathcal{O}(\log_\Delta n)$ rounds deterministically and 
    in $\mathcal{O}(\log_{\delta / r} \delta + \log_{\delta / r} \log n) = \mathcal{O}(\log_\Delta \log n)$ rounds randomized using the algorithms from \Cref{thm:hso} and \Cref{thm:hso_randomized}.
\end{proof}

The orientation of $H$ yields an assignment of matching edges to clusters such that each cluster receives exclusive access to two matching edges. However, the matching edges are not necessarily adjacent to their respective clusters yet. Therefore we modify the matching $M$ by moving each edge closer to the vertex which won the proposal for it.

\begin{lemma} \label{lem:proper_matching}
    There is a matching $M' = \chi(M)$ such that for all $e \in M$ it holds that $\chi(e)$ is adjacent to the vertex that won the proposal for $e$. 
\end{lemma}

\begin{proof}
    For each $e \in M$, let $v$ be the vertex (in the original graph $G$), which sent the winning proposal to $e$.
    If $v$ is incident to $e$, then we simply set $\chi(e) = e$. Otherwise, let $u$ be the endpoint of $e$ which is adjacent to $v$ and set $\chi(e) = uv$.
    With this procedure, we ensure that each cluster is adjacent to at least two edges in $M'$ since each of the two partitions of the cluster receives at least one edge based on the HSO property. 
    If a cluster is assigned more than two edges in $M'$, it discards any edges beyond the required two. 
    In order to show that the edge set $M'$ is still a proper matching, we recall that every vertex $v \in S \subseteq V$ sends out exactly one proposal. 
    Hence, there is no vertex that wins more than one edge in $M$.
    Assume that $v$ won the proposal for the edge $e = uw$. If $v$ is incident to $e$ there is no change from $M$ to $M'$. 
    Otherwise, let $u$ be adjacent to $v$. 
    In this case, the edge $uw \in M$ is flipped to $vw \in M'$.
    Hence, $v$ is the only vertex that gains an incident edge in $M'$ compared to $M$. 
    Since $e$ is one of the closest edges to $v$ in $M$ there cannot be an edge $e' \in M$ that is already incident to $v$.
    Therefore every vertex that wins its only proposal is incident to exactly one edge in $M'$.
    On the other hand, every vertex that does not win its proposal cannot gain any additional adjacent edges in $M'$. 
    Therefore $M'$ is still a matching. 
\end{proof}

\subsection{Phase 3: Switch colors in order to complete the coloring}
\label{subsec:color_switch}

Now that each cluster has been assigned two exclusive edges located at the same distance from its central vertex, we show how to modify their colors to ensure that the colorful condition from \Cref{lem:colorful_leaves} is satisfied for every cluster. To prove that this procedure can be applied to all clusters in parallel, we argue that the colorful condition can still be fulfilled even if the colors of all edges not adjacent to the two exclusive edges are modified adversarially.

\begin{lemma}[Color switching]
    \label{lem:color_switching}
    Let $G = (V,E)$ be a $\Delta$-regular graph, $T = (V_T,E_T) \subseteq G$ a tree, $r \in V_T$ and $\varphi: E \setminus E_T \to [2\Delta - 3]$ a partial edge coloring of $G$. 
    Then, for any two independent edges $uv, u'v' \in E_k$, where $u$ and $u'$ are leaves in $T$, there exists a coloring $\varphi': E \setminus E_T \to [2\Delta - 2]$ such that $|\varphi'(N_E(V_k))| \geq \Delta$ and $\varphi(e) = \varphi'(e)$ for all $e \not\in \{ uv, u'v' \} $.
\end{lemma}

\begin{figure}
    \label{fig:color_flip}
    \centering
    \begin{subfigure}[t]{0.49 \textwidth}
        \centering
        \includegraphics{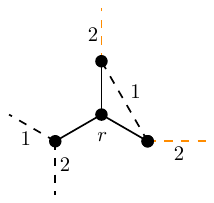}
        \caption{Before switching colors: $\varphi(N_E(V_1)) = \{ 1,2 \} $.}
    \end{subfigure}
    \hfill
    \begin{subfigure}[t]{0.49 \textwidth}
        \centering
        \includegraphics{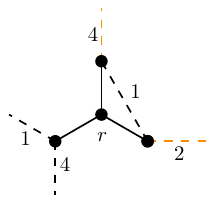}
        \caption{After switching colors: $\varphi(N_E(V_1)) = \{ 1,2,4 \} $.}
    \end{subfigure}
    \caption{
        A simplified example for the color switching procedure, where $\Delta = 3$. 
        The edges in the BFS-tree $T_r$ are represented by solid lines, while edges outside of $T_r$ are dashed. 
        The assigned matching edges $e_1, e_2$ are marked orange.
    }
\end{figure}

\begin{proof}
    Recall that $V_{k} := \{ v \in V_T: \mathrm{dist}_T(v,r) = k \}$. 
    Since $u$ and $u'$ are leaves in $T$, they are incident to exactly $\Delta - 1$ edges in $E \setminus E_T$. 
    Hence, $\varphi(N_E(V_k))$ already contains at least $\Delta - 1$ distinct colors. 
    If $\lvert \varphi(N_E(V_k)) \rvert \geq \Delta $ we simply set $\varphi' = \varphi$.
    Otherwise, if $\lvert \varphi(N_E(V_k)) \rvert = \Delta - 1$, we set $\varphi'(uv) = 2\Delta - 2$. 
    Since $\varphi(e) < 2\Delta - 2$ for all $e \in E_k$ we only need to argue that the color $\varphi(uv)$ still remains in $\varphi'(E_k)$. 
    Consider the sets $N := \{ uw: uw \in E \setminus E_T \} $ and $N' := \{ u'w: u'w \in E \setminus E_T \}$. 
    Since $u, u'$ are both leaves in $T$, both $N$ and $N'$ have cardinality $\Delta - 1$. 
    Moreover, since $\lvert \varphi(E_k) \rvert = \Delta - 1$ and any two edges in $N$ (or $N'$) are adjacent, we have that $\varphi(N) = \varphi(N')$. 
    Hence, there exists an edge $u'v''$ such that $\varphi(uv) = \varphi(u'v'') = \varphi'(u'v'') \in \varphi'(E_k)$. 
    Thus, we have $\varphi'(E_k) = \Delta$. 
    Note that this argument holds even if the colors of the edges not adjacent to $u$ or $v$ are modified adversarially.
\end{proof}

\subsection{Proof of Theorem~\ref{thm:simple_algo}}

\MISalgo*

\begin{proof}
    For $\Delta = \mathcal{O}(1)$ we apply the $(3\Delta/2)$-edge coloring algorithm from \Cref{alg:3Delta/2} to get a $(2\Delta - 2)$-edge coloring in $\mathcal{O}(\log n)$ rounds. 
    Hence, for the remainder of the proof we assume $\Delta \geq \Delta_0 $ for a sufficiently large constant $\Delta_0$. 
    Since the power graph $G^8$ has $n$ vertices and maximum degree polynomial in $\Delta$, an MIS can be computed on it in time $T_\mathrm{MIS}(n, \mathrm{poly}(\Delta))$. 
    Next, the clustering is possible in a constant number of rounds, since the radius of each cluster is at most 8. 
    Also, both computing the $(2\Delta - 3)$-edge coloring $\varphi$ of $G - E_T$ and finding a maximal matching in $E_T$ can be reduced to MIS, and thus can be solved in time $T_\mathrm{MIS}(\Delta^2 \cdot n, \mathrm{poly}(\Delta))$.
    Then, according to \Cref{subsec:HSO} we can modify the matching in a way that each cluster receives two exclusive edges. 
    The rearrangement can be modeled as an instance of HSO with minimum degree $\frac{\Delta^3}{4}$ and maximum rank $\Delta^2$. 
    Since each cluster has constant diameter, simulating this hypergraph incurs just a constant overhead factor.
    Thus, the deterministic HSO algorithm from \Cref{thm:hso} runs in $\mathcal{O}(\log_{\Delta} n)$ rounds. 
    Next, we apply \Cref{lem:color_switching} to ensure that for each cluster there is a $k \leq 8$ such that $\lvert \varphi(N_E(V_k)) \rvert \geq \Delta$.
    By \Cref{lem:colorful_leaves} we can extend $\varphi$ to a proper $(2\Delta - 2)$-edge coloring of $G$.
\end{proof}

\subsection{Replacing MIS and maximal matching (Proof of Theorem~\ref{thm:detAlgo})}
\label{subsec:optimized}

In this subsection we show that computing an MIS is not strictly necessary for our reduction. 
Indeed, using a $(2,k)$-ruling set changes only the diameter of our clusters.
Since increasing the cluster diameter increases the overhead for simulating the auxiliary hypergraph, this induces a tradeoff between $k$ and the runtime of computing a $(2,k)$-ruling set. 
Setting $k = \mathcal{O}(\log \Delta)$ turns out to be the best choice for our purposes.
Secondly, we show how to substitute the maximal matching by a $2$-edge ruling set.
This change is technically more involved, as it requires us to generalize the definition of our auxiliary hypergraph. 
The necessary changes to \Cref{alg:coloring} are also highlighted red in \Cref{alg:optimized}.

\begin{algorithm}
    \caption{Reduction from $(2\Delta-2)$-edge coloring to $(2\Delta - 1)$-edge coloring (high level overview)}
    \label{alg:optimized}
    \begin{algorithmic}[1]
        \State Compute a $(2,\mathcal{O}(\log\Delta))$-ruling set $\cR$ on $G^8$
        {\color{red} \Comment \Cref{cor:det_ruling_set}}
        \State $\mathcal{C}, E_T \gets$ \textsc{Cluster}$(G, \mathcal{R})$
        \Comment \Cref{alg:clustering}
        \State Compute a $(2\Delta-3)$-edge coloring $\varphi$ of $G - E_T$
        \Comment \Cref{lem:2delta-3}
        \State Compute a $2$-edge ruling set $E_\mathcal{R}$ of $T$
        {\color{red} \Comment \Cref{thm:edge-ruling}}
        \State Compute $2$-HSO $\psi: \mathcal{C} \to E_\mathcal{R}$ in $H = (\mathcal{C}, E_\mathcal{R})$
        \Comment \Cref{thm:hso}
        \State Rearrange the matching edges via $\chi: E_\mathcal{R} \to E_\mathcal{R}'$
        \Comment \Cref{lem:proper_matching}
        \For{each cluster $\mathfrak{C} \in \mathcal{C}$ in parallel}
            \State Change the colors of the edges in $\chi(\psi(\mathfrak{C}))$ if necessary
            \Comment{\Cref{lem:color_switching}}
        \EndFor
        \State Extend $\varphi$ to a $(2\Delta-2)$-edge coloring $\varphi^\star$ of $G$
        \Comment \Cref{lem:colorful_leaves}
    \end{algorithmic}
\end{algorithm}

Secondly, we show that we can replace the maximal matching $M$ by a cheaper $2$-edge ruling set $E_\mathcal{R}$ of $G - E_T$. 
In this case, every leaf-vertex of a cluster sees at least one ruling set edge in its $3$-hop neighborhood. 
Here we need to be a little bit more careful in how we set up our hypergraph. 
Firstly, each vertex may only send one proposal to one of its closest ruling set edges. 
Secondly, each vertex may only propose to a ruling set edge $e \in E_\mathcal{R}$, if all the vertices on a shortest path towards this have proposed to $e$ as well. 
For an illustration we refer to \Cref{fig:edge-ruling-set}.

\begin{figure}[ht!]
    \centering
    \begin{subfigure}[t]{0.49 \textwidth}
        \label{fig:edge-ruling-set-before}
        \includegraphics[scale=0.8]{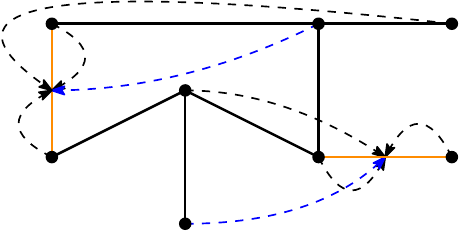}
        \caption{Each vertex proposes to one of the nearest matching edges (orange) within its two-hop neighborhood. Note that vertices lying on the shortest path to a matching edge will propose to the same edge. These proposals are represented by dashed arrows. According to the HSO result, each orange matching edge accepts exactly one proposal, which is indicated by a blue arrow.}
    \end{subfigure}
    \hfill
    \begin{subfigure}[t]{0.49 \textwidth}
        \label{fig:edge-ruling-set-after}
        \includegraphics[scale=0.8]{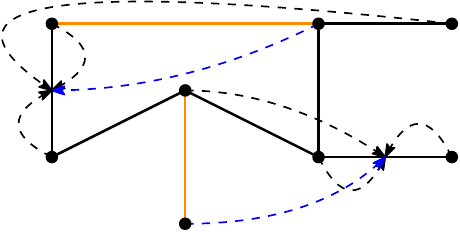}
        \caption{Based on the HSO result (blue arrows), the matching is shifted toward the vertices that won the matching edge. The first edge along the shortest path to these vertices (orange) becomes the new matching edge.}
    \end{subfigure}
    \caption{Illustration of a simplified example showing how the HSO is modeled as a proposal instance for assigning clusters two exclusive edges.}
    \label{fig:edge-ruling-set}
\end{figure}

For each cluster $C \in \mathcal{C}$ we select an arbitrary set $S(C) \subseteq V(C)$ of at least $4 \cdot \Delta^3$ leaves which are all at the same distance from the center.

\begin{definition}[Proposal scheme]
    Each cluster $C \in \mathcal{C}$ sends proposals to matching edges according to the following rules:
    \begin{itemize}
        \item Each vertex $v \in S(C)$ proposes to exactly one matching edge $e \in E_\mathcal{R}$.
        \item If $v$ proposes to $e$, there may not be an edge $e' \in E_\mathcal{R}$ that is closer to $v$ than $e$.
        \item If $v$ proposes to $e$, every vertex on the shortest path from $v$ to $e$ must also propose to $e$.
    \end{itemize}
\end{definition}

\begin{definition}
    Let $e \in E_\mathcal{R}$ and $v \in V$ be the vertex which sent the winning proposal to $e$.
    \begin{itemize} 
        \item If $v$ is incident to $e$, then we simply add $e$ to $M'$.
        \item If $v$ is one hop away from $e$, let $u$ be the endpoint of $e$ which is adjacent to $v$ and add the edge $uv$ to $M'$.
        \item If $v$ is two hops away from $e$, let $u$ be the next vertex on a shortest path from $v$ to $e$ and add the edge $uv$ to $M'$.
    \end{itemize} 
\end{definition}

Now we prove an analogous result to \Cref{lem:proper_matching}.

\begin{lemma}
    The edge set $M'$ is a proper matching in $G - E_T$ and every edge $e \in M'$ is adjacent to the vertex that won the proposal for $e$.
\end{lemma}

\begin{proof}
    Recall that every vertex $v \in V$ sends out at most one proposal. 
    Hence, there is no vertex that wins more than one edge in $E_\mathcal{R}$.
    Assume that $v$ won the proposal for the edge $e = uw \in E_\mathcal{R}$. 
    If $v$ is incident to $e$ there is no change between $E_\mathcal{R}$ and $M'$.
    If $v$ is one hop away from $e$, then $v$ is the only vertex that gains an incident edge in $M'$ compared to $E_\mathcal{R}$. 
    Since $e$ is one of the closest edges to $v$ in $E_\mathcal{R}$ there cannot be an edge $e' \in E_\mathcal{R}$ that is already incident to $v$.
    In the final case that $v$ is two hops away from $e$, both $v$ and the next vertex $u$ on a shortest path from $v$ to $e$ gain an incident edge in $E_\mathcal{R}$. 
    Note that every vertex is at most two hops away from its closest edge in $E_\mathcal{R}$, since $E_\mathcal{R}$ is a 2-edge ruling set.
    By our proposal scheme, $u$ must also have proposed to $v$. 
    Therefore $u$ cannot have won its proposal and $u$ cannot be incident to an edge in $E_\mathcal{R}$. 
    Hence, both $v$ and $u$ are incident to only one edge in $M'$.
    Now assume that there is another vertex $v' = v$ that adds the edge $v'u$ to $M'$. 
    This means that $v'$ must have proposed to the same edge as $u$. 
    But $v$ already won the proposal for this edge, a contradiction!
    Therefore every vertex is incident to at most one edge in $M'$.
\end{proof}

\detAlgo*

\begin{proof}
    The $(2,\mathcal{O}(\log\Delta))$-ruling set can be computed in $\mathcal{O}(\log \Delta + \log^{\ast} n)$ rounds according to \Cref{cor:det_ruling_set}. 
    As an immediate consequence, the diameter of our clusters now increases to $\mathcal{O}(\log \Delta)$. 
    This implies that simulating the hypergraph now incurs a $\mathcal{O}(\log \Delta)$ factor overhead. 
    Therefore the complexity of computing a HSO increases to $\mathcal{O}(\log_\Delta n \cdot \log \Delta) = \mathcal{O}(\log n)$ rounds. 
    The $(2\Delta - 3)$-edge coloring of $G - E_T$ remains unchanged. 
    Further, we can also relax the requirements for our matching $M$ to $2$-edge ruling set, which can be computed in $\mathcal{O}(\log^{\ast} n)$ rounds.
    Therefore the total runtime of our reduction works out to be 
    \begin{equation*}
        \mathcal{O}(\log \Delta + \log^{\ast} n) + \mathcal{O}(\log \Delta) + T_{2\Delta - 1}(n) + \mathcal{O}(\log^{\ast} n) + \mathcal{O}(\log \Delta \cdot \log n) = \mathcal{O}(\log n) + T_{2\Delta - 1}(n). \qedhere
    \end{equation*}
\end{proof}

\detAlgoCor*

\begin{proof}
    For the first claim we plug in the $\mathcal{O}(\log^{12} \Delta + \log^\ast n)$-round algorithm from \Cref{alg:polylog} into \Cref{thm:detAlgo}.
    For the latter we plug in the $\widetilde{\mathcal{O}}(\log^{5/3} n)$-round MIS algorithm from \Cref{alg:greedy} into \Cref{thm:simple_algo} instead.
\end{proof}

\section{Randomized \texorpdfstring{$(2\Delta - 2)$}{(2Δ-2)}-edge coloring algorithm}
\label{sec:randomized}

In order to achieve a $\mathcal{O}(\log \log n)$-round randomized reduction, we only need to replace the ruling set algorithm by a faster randomized variant in \Cref{cor:rand_ruling_set} and the deterministic HSO algorithm by its randomized counterpart in \Cref{thm:hso_randomized}. The changes to the optimized deterministic version are also highlighted red in \Cref{alg:randomized}.

\begin{algorithm}
    \caption{Randomized $(2\Delta-2)$-edge coloring (high level overview)}
    \label{alg:randomized}
    \begin{algorithmic}[1]
        \State Compute a $(2, \mathcal{O}(\log \Delta) \cap \mathcal{O}(\log \log n))$-ruling set $\cR$ on $G^8$
        {\color{red}\Comment \Cref{cor:rand_ruling_set}}
        \State $\mathcal{C}, E_T \gets$ \textsc{Cluster}$(G, \mathcal{R})$
        \Comment \Cref{alg:clustering}
        \State Compute a $(2\Delta-3)$-edge coloring $\varphi$ of $G - E_T$
        \Comment \Cref{lem:2delta-3}
        \State Compute a $2$-edge ruling set $E_\mathcal{R}$ of $T$
        \Comment \Cref{thm:edge-ruling}
        \State Compute $2$-HSO $\psi: \mathcal{C} \to E_\mathcal{R}$ in $H = (\mathcal{C}, E_\mathcal{R})$ 
        {\color{red}\Comment \Cref{thm:hso_randomized}}
        \State Rearrange the matching edges via $\chi: E_\mathcal{R} \to E_\mathcal{R}'$
        \Comment \Cref{lem:proper_matching}
        \For{each cluster $\mathfrak{C} \in \mathcal{C}$ in parallel}
            \State Change the colors of the edges in $\chi(\psi(\mathfrak{C}))$ if necessary
            \Comment{\Cref{lem:color_switching}}
        \EndFor
        \State Extend $\varphi$ to a $(2\Delta-2)$-edge coloring of $G$
        \Comment \Cref{lem:colorful_leaves}
    \end{algorithmic}
\end{algorithm}

\randAlgo*

\begin{proof}
    The $(2, \mathcal{O}(\log \Delta) \cap \mathcal{O}(\log \log n))$-ruling set can be computed in $\mathcal{O}(\log \log n)$ rounds according to \Cref{cor:rand_ruling_set}. 
    Hence, each cluster has diameter at most $\mathcal{O}(\log \log n)$ and can be computed in $\mathcal{O}(\log \log n)$ rounds as well. 
    For the $2$-edge ruling set the deterministic algorithm from \Cref{thm:edge-ruling} running in $\mathcal{O}(\log^{\ast} n)$ rounds is still fast enough for our purposes here. 
    According to \Cref{thm:hso_randomized} and \Cref{subsec:HSO} the hypergraph sinkless orientation can be computed in $\mathcal{O}(\log_\Delta \log n)$ rounds on the cluster graph, where we incur a $\mathcal{O}(\log \Delta)$-round simulation overhead. 
    Finally, the extension to the edges inside each cluster can be computed in $\mathcal{O}(\log \log n)$ rounds. 
    Hence, the total runtime works out to be
    \begin{equation*}
        \mathcal{O}(\log \log n) + T_{2\Delta - 1}^{\dice{6}}(n) + \mathcal{O}(\log^\ast n) + \mathcal{O}(\log \Delta \cdot \log_\Delta \log n) = \mathcal{O}(\log \log n) + T_{2\Delta - 1}^{\dice{6}}(n). \qedhere
    \end{equation*}
\end{proof}

\randAlgoCor*

\begin{proof}
    To prove the first claim we plug in the deterministic $\mathcal{O}(\log^{12}\Delta + \log^{\ast} n)$-round algorithm from \Cref{alg:polylog} into \Cref{thm:randAlgo}, while the second claim uses the randomized $\widetilde{\mathcal{O}}(\log^{5/3} \log n)$-round $(2\Delta - 1)$-edge coloring algorithm from \Cref{alg:greedy_randomized}.
\end{proof}

Similar to the deterministic case we could also get a $\mathcal{O}(\log_\Delta \log n)$-round reduction to MIS via replacing the $(2, \mathcal{O}(\log \Delta) \cap \mathcal{O}(\log \log n))$-ruling set by an MIS. 
However, applying the shattering framework of \cite{ShatteringMIS} to the deterministic MIS algorithm from \cite{GreedyAlgorithm} still takes $\mathcal{O}(\log \Delta) + \widetilde{\mathcal{O}}(\log^{8/3} \log n)$ rounds. 
Hence, the additional cost of computing an MIS currently outweighs the time saved from the faster reduction.

\appendix

\section{Complementary subroutines}

All subroutines from the existing literature used throughout our algorithms are stated here explicitly.

\subsection{Ruling sets}

In this subsection we collect and combine various ruling set algorithms from the distributed literature that we will use as subroutines for our edge coloring algorithm.

\begin{theorem}[{ruling set algorithm \cite[Theorem 8]{RulingSetBoundedGrowth}}]
    \label{thm:ruling_set_induced_degree}
    There is a randomized distributed algorithm that computes a $(1,\mathcal{O}(\log\log\Delta))$-ruling set with induced degree at most $\mathcal{O}(\log^{5} n)$ in any $n$-vertex graph with maximum degree $\Delta$ in $\mathcal{O}(\log\log\Delta)$ rounds
\end{theorem}

\begin{theorem}[{ruling set algorithm \cite[Theorem 3]{RulingSetColoring}}]
    \label{thm:ruling_set_coloring}
    Let $G = (V,E)$ and $W \subseteq V$. Given a $d$-coloring of $G$, there is a deterministic distributed algorithm that computes a $(2,c)$-ruling set for $W$ in time $\mathcal{O}(c \cdot d^{1/c})$.
\end{theorem}

\begin{corollary}
    \label{cor:det_ruling_set}
    There is a deterministic distributed algorithm that computes a $(2,\mathcal{O}(\log \Delta))$-ruling set in any graph $G$ in $\mathcal{O}(\log \Delta + \log^{\ast} n)$ rounds of the \LOCAL model.
\end{corollary}

\begin{proof}
    We start by computing a $\mathcal{O}(\Delta^2)$-vertex coloring $\varphi$ of $G$ in $\mathcal{O}(\log^\star n)$ rounds using Linial's algorithm \cite{Linial1, Linial2}. 
    Then, we apply \Cref{thm:ruling_set_coloring} to $\varphi$ for $c = \mathcal{O}(\log \Delta)$ to get a $(2,\mathcal{O}(\log \Delta))$-ruling set in time $\mathcal{O}(c \cdot d^{1/c}) = \mathcal{O}(\log \Delta \cdot \Delta^{2/\log \Delta}) = \mathcal{O}(\log \Delta)$.
\end{proof}

\begin{corollary}
    \label{cor:rand_ruling_set}
    There is a randomized distributed algorithm that computes a $(2, \mathcal{O}(\log \Delta) \cap \mathcal{O}(\log \log n))$-ruling set in any graph $G$ in $\mathcal{O}(\log \log n)$ rounds of the \LOCAL model.
\end{corollary}

\begin{proof}
    If $\Delta = \mathcal{O}(\log n)$, we can simply use the deterministic algorithm from \Cref{cor:det_ruling_set} to get an $(2,\mathcal{O}(\log \Delta))$-ruling set (which is also a $(2,\mathcal{O}(\log \log n))$-ruling set) in $\mathcal{O}(\log \Delta + \log^{\ast} n) = \mathcal{O}(\log \log n)$ rounds. 
    Otherwise, for $\Delta = \Omega(\log n)$ we first apply \Cref{thm:ruling_set_induced_degree} to get a $(1,\mathcal{O}(\log \log \Delta))$-ruling set $S_0$ with induced degree at most $\mathcal{O}(\log^5 n)$. 
    Next we compute a $\mathcal{O}(\log^{10} n)$-coloring $\chi$ of $G[S_0]$ using Linial's algorithm \cite{Linial1, Linial2}. 
    Finally using $\chi$ as the input coloring and setting $c = \mathcal{O}(\log \log n)$, we get a $(2,\mathcal{O}(\log \log n))$ ruling set $S_1$ for $S_0$ in 
    $$
    \mathcal{O}(c \cdot d^{1/c}) = \mathcal{O}(\log \log n \cdot \log^{10 / \log \log n} n) = \mathcal{O}(\log \log n).
    $$
    randomized time. 
    Since every node $v \in V$ is at most $\mathcal{O}(\log \log \Delta) = \mathcal{O}(\log \log n)$ hops away from a vertex in $S_0$ and every vertex in $S_0$ is at most $\mathcal{O}(\log \log n) = \mathcal{O}(\log \Delta)$ hops away from a vertex in $S_1$, we get that $S_1$ is both a $(2,\mathcal{O}(\log \Delta))$-ruling set and a $(2,\mathcal{O}(\log \log n))$-ruling set for $V$. 
\end{proof}

Another subroutine that we will use in the optimized version of our algorithm computes a relaxation of a maximal matching.

\begin{theorem}[{$2$-ruling edge set algorithm \cite[Theorem 1.3]{RulingLineGraphs}}]
    \label{thm:edge-ruling}
    There is a deterministic distributed algorithm that computes a $2$-ruling edge set of $G$ in $\Theta(\log^\star n)$ rounds in the \CONGEST model.
\end{theorem}

\subsection{Hypergraph sinkless orientation}

\begin{theorem}[{deterministic HSO \cite[Theorem 1.4]{HSO}}]
    \label{thm:hso}
    There is a deterministic distributed algorithm for computing an HSO in any $n$-vertex multihypergraph $H = (V,E)$ with maximum rank $r$ and minimum degree $\delta > r$ in $\mathcal{O}(\log_{\delta / r} n)$ rounds.
\end{theorem}

\begin{theorem}[{randomized HSO \cite[Theorem 1.5]{HSO}}]
    \label{thm:hso_randomized}
    There is a randomized distributed algorithm that w.h.p. computes an HSO on any hypergraph with maximum rank $r$ and minimum degree $\delta \geq 320 r \log r$ with runtime $\mathcal{O}(\log_{\delta/r} \delta + \log_{\delta/r}\log n)$.
\end{theorem}

\subsection{Maximal Independent Set}

\begin{theorem}[deterministic MIS {\cite[Theorem 1.3]{GreedyAlgorithm}}]
    \label{alg:greedy}
    There is a deterministic distributed algorithm that in any $n$-vertex graph $G = (V,E)$ computes an MIS in $\widetilde{\mathcal{O}}(\log^{5/3} n)$ rounds of the \LOCAL model. 
    This complexity also applies for maximal matching, $(\mathrm{deg} +1)$-list vertex coloring, and $(2\cdot \mathrm{deg} -1)$-list edge coloring.
\end{theorem}

\begin{corollary}[randomized MIS {\cite[Theorem 1.3]{GreedyAlgorithm} \& \cite{OptimalColoring}}]
    \label{alg:greedy_randomized}
    There is a randomized distributed algorithm, in the \LOCAL model, that computes a $(\Delta + 1)$-vertex coloring in $\widetilde{\mathcal{O}}(\log^{5/3}\log n)$ rounds. 
    The same holds also for the $(2\Delta - 1)$-edge coloring problem.
\end{corollary}

\subsection{Edge coloring}

\begin{theorem}[deterministic $(2\Delta - 1)$-edge coloring {\cite[Theorem D.4]{EdgeColoringPolylogarithmicInDelta}}]
    \label{alg:polylog}
    The $(\mathrm{deg}(e) + 1)$-list edge coloring problem can be solved in $\mathcal{O}(\log^7 C \cdot \log^5 n + \log^\ast n)$ deterministic rounds in the \LOCAL model, where $C$ is the size of the color space.
\end{theorem}

\begin{theorem}[deterministic $(3\Delta/2)$-edge coloring {\cite[Theorem 1.2]{HSO}}]
    \label{alg:3Delta/2}
    There is a deterministic $\mathcal{O}(\Delta^2 \cdot \log n)$-round \LOCAL algorithm that computes a $(3\Delta / 2)$-edge coloring on any $n$-vertex graph with maximum degree $\Delta$.
\end{theorem}

\bibliographystyle{plainurl}
\bibliography{references}

\end{document}